\newif\ifTAC

\TACfalse

\documentclass[journal,twoside,web]{ieeecolor}
\usepackage{generic}

\usepackage{textcomp}
\usepackage{balance}
\usepackage{algorithm}
\usepackage{algpseudocode}
\usepackage{enumerate}
\usepackage{amsmath} 
\usepackage{amssymb}  
\usepackage{graphicx} 
\usepackage{amsfonts}
\usepackage{cite}
\usepackage[font=small,skip=1pt]{subcaption} 
\usepackage[font=small,skip=1pt]{caption}
\usepackage{url}
\usepackage{booktabs}
\usepackage{bm}

\usepackage{soul}  
\usepackage{comment} 
\usepackage{mathtools} 

\usepackage{amsthm} 
\usepackage[utf8]{inputenc}
\usepackage[english]{babel}

\usepackage{siunitx}    

\newtheorem{remark}{Remark}
\newtheorem{theorem}{Theorem}
\newtheorem{proposition}{Proposition}
\newtheorem{definition}{Definition}[section]
\newtheorem{lemma}{Lemma}
\newtheorem{assumption}{Assumption}

\usepackage{ifthen} 

\newboolean{conference}
\setboolean{conference}{true}  

\DeclareMathOperator*{\diag}{diag}
\DeclareMathOperator*{\blkdiag}{blkdiag}

\DeclareMathOperator*{\argmin}{\arg\!\min}

\DeclareMathOperator*{\col}{col}

\DeclareMathOperator*{\fix}{Fix}
\newcommand{\dist}[1]{\ensuremath{\mathrm{d}_{#1}}}

\newcommand{\nR}[1]{\mathbb{R}^{#1}}		

\newcommand{\vect}[1]{\ensuremath{\bm{#1}}}		
\newcommand{\matr}[1]{\ensuremath{\bm{#1}}}		

\newcommand{\define}{\coloneqq}			
\newcommand{\norm}[1]{\left\lVert#1\right\rVert}

\newcommand{\operator}[1]{\mathrm{#1}}
\newcommand{\ideop}{\operator{Id}}

\newcommand{\verticalstack}[4]{\col(#1_{#2})_{#3 \in #4}}
\newcommand{\blkdiagconcat}[4]{\blkdiag(#1_{#2})_{#3 \in #4}}

\newcommand{\bPhi}{\boldsymbol{\Phi}}
\newcommand{\bGamma}{\boldsymbol{\Gamma}}

\newcommand{\eye}[1]{\matr{I}_{#1}}
\newcommand{\zeros}[1]{\matr{0}_{#1}}
\newcommand{\ones}[1]{\matr{1}_{#1}}

\newcommand{\pos}{\vect{p}}				
\newcommand{\vel}{\vect{v}}				

\newcommand{\state}{\boldsymbol{x}}
\newcommand{\sysinput}{\boldsymbol{u}}

\newcommand{\Nrob}{\ensuremath{N}}

\newcommand{\graph}{\ensuremath{\mathcal{G}}}
\newcommand{\edges}{\ensuremath{\mathcal{E}}}
\newcommand{\vertices}{\ensuremath{\mathcal{V}}}
\newcommand{\neigh}{\ensuremath{\mathcal{N}}}

\newcommand{\stateest}{\hat{\state}}

\newcommand{\meas}{\boldsymbol{y}}

\newcommand{\measmat}{\boldsymbol{H}}

\newcommand{\statemat}{\boldsymbol{A}}
\newcommand{\inputmat}{\boldsymbol{B}}

\newcommand{\statetrans}{\bPhi}
\newcommand{\obsgram}{\boldsymbol{G}}

\newcommand{\ffmatrix}{\bGamma}

\newcommand{\infovec}{\boldsymbol{z}}
\newcommand{\infomat}{\boldsymbol{S}}

\newcommand{\covmat}{\boldsymbol{P}}
\newcommand{\outcov}{\boldsymbol{R}}

\newcommand{\ffactor}{\gamma}
\newcommand{\correction}{\boldsymbol{\xi}}
\newcommand{\estcorrection}{\hat{\correction}}

\newcommand{\procnoisemat}{\boldsymbol{Q}}

\newcommand{\knownterm}{\boldsymbol{b}}

\newcommand{\localcorrection}[2]{\correction_{#1}^{(#2)}}
\newcommand{\extloccorr}[1]{\correction_{\neigh_{#1}}}
\newcommand{\estextloccorr}[2]{\hat{\correction}_{\neigh_{#1}, #2}}
\newcommand{\estlocalcorrection}[2]{\hat{\correction}_{#1}^{(#2)}}

\newcommand{\iteradmm}{{h}}
\newcommand{\maxiteradmm}{{H}}

\newcommand{\frozeniteradmm}{{{k}}}

\newcommand{\estate}{\tilde{\state}}
\newcommand{\ecorr}{\tilde{\correction}}

\newcommand{\smallgain}{{\varepsilon}}

\title{\LARGE \bf 
ADMM-Based Distributed Kalman-like Observer \\ with Applications to Cooperative Localization}

\author{
Nicola De Carli, Nicola Bastianello, Dimos V. Dimarogonas
\thanks{
This work was supported in part by the Wallenberg AI, Autonomous Systems and Software Program (WASP) funded by the 
Knut and Alice Wallenberg (KAW) Foundation, the ERC LEAFHOUND Project, the Swedish Research Council (VR), Digital Futures and the Swedish Foundation for Strategic Research project.
}
\thanks{
Nicola De Carli, Nicola Bastianello and Dimos V. Dimarogonas are with the Department of Decision and Control Systems, KTH Royal Institute of Technology, Stockholm, Sweden. E-mail:
\{ndc,nicolba,dimos\}@kth.se 
}
}

\ifTAC
    \def\BibTeX{{\rm B\kern-.05em{\sc i\kern-.025em b}\kern-.08em
        T\kern-.1667em\lower.7ex\hbox{E}\kern-.125emX}}
\else
    \def\BibTeX{{\rm B\kern-.05em{\sc i\kern-.025em b}\kern-.08em
        T\kern-.1667em\lower.7ex\hbox{E}\kern-.125emX}}
\fi

\begin{document}

\maketitle

\begin{abstract}
This paper addresses distributed state estimation for multi-agent systems with local and relative measurements, motivated by cooperative localization problems in which the global state dimension scales with the size of the network. We consider a Kalman-like observer in information form and introduce a sparsity-preserving prediction step based on an exponential forgetting factor, thereby avoiding the dense Riccati recursion of the standard information filter. The correction step is recast as a strongly convex quadratic program with structure induced by the sensing graph, which enables a distributed solution based on the \emph{alternating direction method of multipliers} (ADMM). In the resulting scheme, each agent updates local copies of its own correction variable and those of its neighbors using only local communication, thus avoiding centralized matrix inversion and consensus over full global-state quantities. A two-time-scale stability analysis is developed for the interconnected observer: the reduced estimation-error dynamics are shown to be uniformly exponentially stable, the ADMM dynamics define an exponentially stable fast subsystem, and these properties are combined to establish uniform exponential stability of the overall distributed observer. Numerical simulations in a multi-agent cooperative localization scenario illustrate the performance of the proposed distributed observer.
\end{abstract}

\begin{IEEEkeywords}
Multi-agent systems, distributed optimization, cooperative localization, observer design
\end{IEEEkeywords}

\section{Introduction}
Distributed cooperative state estimation is a fundamental problem in networked systems, with applications including cooperative localization \cite{luft2018recursive, kia2016cooperative, roumeliotis2002distributed, chang2021resilient, de2024adaptive} and distributed power network state-estimation \cite{gomez2011taxonomy, yang2021distributed}. In such scenarios, agents may have access to local measurements, depending only on their own state, as well as relative measurements involving neighboring agents. It could be the case, however, that local measurements are available only to a subset of the agents, while the others must estimate their state through relative sensing and information exchange. While these sensing modalities provide rich information, they also induce estimation problems whose dimension grows with the size of the network, making centralized approaches increasingly difficult to scale.


Distributed estimation strategies are thus particularly appealing. The Kalman filter is the standard solution for state estimation, and its information form is well suited to distributed sensing since measurement information is aggregated additively~\cite{battistelli2014consensus, kamal2013information, sebastian2023eco}.

A large body of work has addressed distributed Kalman filtering and information fusion in sensor networks \cite{battistelli2014consensus, kamal2013information, sebastian2023eco}, often relying on consensus mechanisms to aggregate information terms. These approaches are effective when the network estimates a global state of fixed dimension, as in target tracking \cite{kamal2013information}. In contrast, in cooperative localization and related multi-agent estimation problems, the state dimension grows with the number of agents, while each robot typically needs an accurate estimate only of its own state and, at most, of those of nearby agents. As a result, reconstructing the full state of the entire network may be unnecessary and computationally inefficient.
In this setting, the direct use of the information filter faces two main limitations. First, the standard prediction step generally destroys the sparsity of the information matrix, since the Riccati recursion introduces dense couplings even when the underlying sensing interactions are local. Second, the correction step requires the inversion of a global information matrix, which is incompatible with fully distributed implementations and scales poorly with the number of agents.

In the cooperative localization literature, a number of approaches have been proposed to distribute Kalman filtering while coping with the dense covariance matrix and the cross-correlations induced by inter-agent measurements
\cite{luft2018recursive,kia2016cooperative,jung2020decentralized,roumeliotis2002distributed,chang2021resilient}. Broadly speaking, these methods can be grouped into three main categories: $(i)$ centralized-equivalent schemes
\cite{roumeliotis2002distributed,kia2016cooperative}, $(ii)$ approximation-based schemes that modify the covariance propagation or update
\cite{luft2018recursive,jung2020decentralized,roumeliotis2002distributed}, and $(iii)$ covariance-intersection (CI) based methods, which are designed to fuse estimates under unknown correlations
\cite{arambel2001covariance,carrillo2013decentralized,klingner2019fault,chang2021resilient}.

Centralized-equivalent approaches
\cite{roumeliotis2002distributed,kia2016cooperative}
aim at reproducing exactly the estimate that would be obtained by a centralized Kalman filter. Their main advantage is that they preserve the optimality and consistency properties of the centralized estimator. However, this is typically achieved at the expense of substantial communication and bookkeeping requirements, often involving all-to-all information exchange or the propagation of cross-covariance terms across the network. As a result, these methods become increasingly difficult to scale as the number of agents grows.

A second line of work seeks scalability by approximating the covariance recursion. In particular, the method in \cite{luft2018recursive}, on which \cite{jung2020decentralized} further builds, introduces approximations in the prediction and correction steps so as to preserve a distributed structure while remaining close, in a suitable sense, to the nominal Riccati iteration. These methods have been shown to yield good empirical performance in cooperative localization problems
\cite{luft2018recursive,jung2020decentralized,de2021online},
and they offer an attractive compromise between estimation quality and computational tractability. Nevertheless, because the exact covariance recursion is modified, the resulting estimates are in general not consistent, and rigorous stability properties of the overall observer are typically not established.

A third class of approaches relies on covariance intersection and related conservative fusion rules
\cite{arambel2001covariance,carrillo2013decentralized,klingner2019fault,chang2021resilient}. These methods are particularly appealing when the cross-correlations among local estimates are unknown or difficult to track, since they allow multiple estimates to be fused while preserving consistency. Their main drawback, however, is conservatism: to guarantee consistency without explicit correlation information, they often underuse the information provided by neighboring agents, which may lead to degraded estimation accuracy.

Overall, most existing distributed cooperative localization methods primarily focus on preventing overconfidence when fusing information across agents, or on approximating the centralized covariance recursion in a scalable way. By contrast, the question of establishing stability of the resulting distributed observer is often left open. This gap is one of the main motivations for the present work, which seeks to combine a fully distributed implementation with a stability guarantee while still exploiting information from all agents, so as to asymptotically recover the corresponding centralized Kalman-like observer.

In our previous work \cite{de2025distributed}, building on the Kalman-like design of \cite{bernard2020semi}, we proposed a continuous-time distributed Kalman-like observer in which the process-noise term is replaced by a forgetting factor, thereby preserving the sparsity of the information matrix, while the correction step is computed through a distributed dynamic inversion scheme. That approach relies on a continuous-time version of the Richardson iteration \cite{bertsekas2015parallel}, which can be used to solve sparse positive-definite linear systems in a distributed manner \cite{de2024distributed}. Although this method can provide good estimation performance, its convergence may become slow when the matrix to be inverted is ill-conditioned. Moreover, in the discrete-time setting, the admissible Richardson step size depends on spectral properties of the time-varying information matrix. In practice, distributed implementations typically require a fixed step size that remains valid over time; however, deriving a sufficiently tight offline bound may be conservative, while online adaptation is difficult to realize in a fully distributed way. 

To address these limitations, this paper proposes a distributed discrete-time Kalman-like observer that integrates a sparsity-preserving prediction step with an optimization-based correction stage. By employing an exponential forgetting factor \cite{cticlea2013exponential,bernard2020semi}, the prediction maintains the sparsity of the sensing topology and avoids the computational burden of standard Riccati recursions. The correction step is then recast as a structured quadratic program and solved via a partition-based ADMM \cite{neal2011distributed, bastianello2018partition}. This architecture allows agents to maintain and update local copies of neighboring variables using only one-hop communication, effectively removing the need for centralized matrix inversions. Notably, the framework enables the closed-form integration of local measurements, reserving the distributed optimization layer exclusively for the residual correction induced by relative sensing. This allows agents to exploit private information immediately while minimizing the communication overhead associated with the collaborative estimation task.

The main contributions of this paper are:
\begin{enumerate}
\item The formulation of a discrete-time Kalman-like observer in information form that preserves the sparsity of the sensing graph during prediction, bypassing the dense Riccati recursions typical of standard information filters.

\item A distributed correction architecture based on a partition-based ADMM. This scheme asymptotically recovers the centralized estimate while requiring only local communication, avoids consensus over the full global state, and allows for the direct, closed-form integration of local sensing.

\item A rigorous stability analysis of the resulting interconnected observer. Using a two-time-scale framework and small gain arguments, we establish the uniform global exponential stability (UGES) of the full system by showing that the fast ADMM dynamics track the slow estimation-error dynamics.
\end{enumerate}

\paragraph*{Notation}
Scalars, vectors, matrices, and sets are denoted by lowercase letters, bold lowercase letters, bold uppercase letters, and calligraphic symbols, respectively. The identity and zero matrices are denoted by $\eye{n}$ and $\zeros{n}$, with dimensions omitted when clear from the context. The operator $\col(\cdot)$ denotes vertical concatenation. In particular, given a finite index set $\vertices$, the expression $\col(\state_{i,k})_{i\in\vertices}$ denotes the vector obtained by stacking the blocks $\state_{i,k}$ for all $i\in\vertices$. For instance, if $\vertices=\{1,\dots,\Nrob\}$, then
{\small $
\col(\state_{i,k})_{i\in\vertices}
=
\begin{bmatrix}
\state_{1,k}^\top & \cdots & \state_{\Nrob,k}^\top
\end{bmatrix}^\top.
$}
Likewise, $\blkdiag(\cdot)$ denotes block-diagonal concatenation, so that $\blkdiag(\boldsymbol{A}_i)_{i\in\vertices}$ is the block-diagonal matrix whose diagonal blocks are the matrices $\boldsymbol{A}_i$, arranged according to the same ordering of $\vertices$. The symbol $\otimes$ denotes the Kronecker product. For symmetric matrices $\boldsymbol{A}$ and $\boldsymbol{B}$, $\boldsymbol{A}\succeq \boldsymbol{B}$ means that $\boldsymbol{A}-\boldsymbol{B}$ is positive semidefinite. For $\boldsymbol{M}\succ \boldsymbol{0}$, $\norm{\boldsymbol{x}}_{\boldsymbol{M}}^2:=\boldsymbol{x}^\top \boldsymbol{M}\boldsymbol{x}$ denotes the weighted squared norm. Finally, $\lambda_{\max}(\boldsymbol{A})$ denotes the maximum eigenvalue of $\boldsymbol{A}$.

\section{System Model}
Consider a network of $\Nrob$ agents with decoupled discrete-time linear dynamics
{\small
\begin{equation}
    \state_{i,k+1}
    =
    \statemat_{i,k}\state_{i,k}
    +
    \inputmat_{i,k}\sysinput_{i,k},
    \qquad i \in \vertices := \{1,\dots,\Nrob\},
    \label{eq:local_dyn}
\end{equation}
}%
where $\state_{i,k} \in \nR{d}$, $\sysinput_{i,k} \in \nR{m}$, and the matrices
$\statemat_{i,k} \in \nR{d \times d}$ and
$\inputmat_{i,k} \in \nR{d \times m}$ may be time-varying.\footnote{The extension to heterogeneous agent dimensions is straightforward and omitted for brevity. Similarly, for heterogeneous measurements.}



Agent interactions are described by: $(i)$ a directed \emph{sensing graph} $\graph_s = (\vertices,\edges_s)$, and $(ii)$ an undirected \emph{communication graph} $\graph_c = (\vertices,\edges_c)$.
%
A directed edge $(i,j)\in\edges_s$ indicates that agent $i$ obtains a relative measurement involving agent $j$. The communication graph is assumed to be the undirected counterpart
of the sensing graph, i.e. if $(i,j)\in\edges_s$ or $(j,i)\in\edges_s$, then $\{i,j\}\in\edges_c$,
ensuring bidirectional information exchange.
We denote the set of communication neighbors of agent $i$ by
$
\neigh_i := \{j \in \vertices \mid \{i,j\}\in\edges_c\}.
$



Each agent may obtain: $(i)$ a \emph{local measurement} ${\meas_{i,k}^{\ell} \in \nR{q_{\ell}}}$
{\small
\begin{equation}
        \meas_{i,k}^{\ell}
        =
        \measmat_{i,k}^{\ell} \state_{i,k} \in \nR{q_{\ell}},
        \qquad i \in \vertices_{\ell} \subseteq \vertices,
        \label{eq:private_meas}
    \end{equation}
    }%
$(ii)$ a \emph{relative measurement} $\meas_{ij,k}^r \in \nR{q_r}$ 
for each $(i,j)\in\edges_s$
{\small
\begin{equation}
    \meas_{ij,k}^r
    =
    \measmat_{ij,i,k}^{r} \state_{i,k}
    +
    \measmat_{ij,j,k}^{r} \state_{j,k} \in \nR{q_r}.
    \label{eq:relative_meas}
\end{equation}
}%
Here $\vertices_{\ell}$ denotes the subset of agents equipped with sensors providing local measurements.



Define the stacked state and input vectors
$\state_k := \verticalstack{\state}{i}{k}{\vertices}$, and
$\sysinput_k := \verticalstack{\sysinput}{i,k}{i}{\vertices}$.
%
Let
$\statemat_k := \blkdiagconcat{\statemat}{i,k}{i}{\vertices}$, and
$\inputmat_k := \blkdiagconcat{\inputmat}{i,k}{i}{\vertices}$.

Then the network dynamics compactly read
{\small
\begin{equation}
    \state_{k+1}
    =
    \statemat_k \state_k
    +
    \inputmat_k \sysinput_k.
    \label{eq:global_dyn}
\end{equation}
}%
The stacked measurement vector is
{\small
\begin{equation}
\meas_k :=
\begin{bmatrix}
\verticalstack{\meas^\ell}{i,k}{i}{\vertices_{\ell}} \\
\verticalstack{\meas^r}{ij,k}{(i,j)}{\edges_{s}}
\end{bmatrix}= \measmat_k \state_k.
\label{eq:global_meas}
\end{equation}
}%
Here, $\measmat_k$ is partitioned into block rows associated with individual measurements and block columns associated with agents. Its sparsity reflects the measurement structure: each block row has nonzero blocks only in the columns corresponding to the agents involved in that measurement. Accordingly, a local measurement of agent $i$ contributes a single nonzero block in column $i$, while a relative measurement between agents $i$ and $j$ contributes nonzero blocks only in columns $i$ and $j$. Hence, $\measmat_k$ has the sparsity pattern of a block incidence matrix of the measurement graph.


The subsequent analysis relies on the following standard assumptions \cite{cticlea2013exponential} for time-varying linear systems and Kalman-type observers and additional bounds on the variation of the system matrices.

\begin{assumption}[Uniform boundedness]
\label{ass:boundedness}
There exist constants $\bar a,\bar b,\bar h,\bar a_{\Delta}, \bar h_\Delta>0$ such that, for all $k$,
{\small
\[
\begin{aligned}
    &\|\statemat_k\|\le \bar a,
\qquad
\|\inputmat_k \sysinput_k\|\le \bar b,
\qquad 
\|\measmat_k\|\le \bar h
\\
&\|\statemat_{k+1} - \statemat_{k}\|\le \bar a_{\Delta}
\qquad
\|\measmat_{k+1} - \measmat_{k}\|\le \bar h_{\Delta}.
\end{aligned}
\]
}%
\end{assumption}

As in standard Kalman filtering, an appropriate uniform observability condition is
required. To this end, let $\statetrans(k,k_0)$ denote the state transition
matrix of \eqref{eq:global_dyn}, namely
{\small
\[
\statetrans(k,k_0)
\define
\begin{cases}
\statemat_{k-1}\statemat_{k-2}\cdots \statemat_{k_0}, & k>k_0,\\[1mm]
\eye{\Nrob d}, & k=k_0.
\end{cases}
\]
}%
Given an integer $K\geq 1$, define the observability Gramian over $K$ samples as
{\small
\begin{equation}
\obsgram(k,K)
\define
\sum_{t=k-K}^{k}
\statetrans(t,k-K)^\top
\measmat_t^\top \measmat_t
\statetrans(t,k-K).
\label{eq:obs_gramian}
\end{equation}
}%
\begin{assumption}[Complete uniform observability]
\label{ass:compobs}
There exist an integer $K\geq 1$ and constants $\alpha_1,\alpha_2>0$ such that,
for all $k\geq K$,
{\small
\[
\alpha_1 \eye{\Nrob d}
\preceq
\obsgram(k,K)
\preceq
\alpha_2 \eye{\Nrob d}.
\]
}%
\end{assumption}

Finally, since the proposed information prediction step involves
$\statemat_k^{-1}$, we require the state-transition matrices to be uniformly
invertible.

\begin{assumption}[Uniform invertibility]
\label{ass:invertibility}
For all $k$, the matrix $\statemat_k$ is invertible, and there exists a constant
$\bar a_{\mathrm{inv}}>0$ such that
\[
\|\statemat_k^{-1}\| \le \bar a_{\mathrm{inv}}, \qquad \forall k.
\]

\end{assumption}
%

\section{Information-Form Observer and Distributed Correction}
\label{sec:observer}
%
This section develops the proposed distributed Kalman-like observer in information form. We first exploit the additive measurement-update structure to obtain a graph-structured information update and combine it with a sparsity-preserving prediction step based on exponential forgetting. We then rewrite the correction step as a sparse symmetric positive definite linear system, or equivalently as a structured strongly convex quadratic program, which enables the distributed ADMM-based implementation developed later in the section.

Let $\covmat_{k|k-1}$ and $\stateest_{k|k-1}$ denote the prior covariance and state estimate, and define the corresponding prior information matrix and vector as
{\small
\begin{equation}
    \infomat_{k|k-1} := \covmat_{k|k-1}^{-1},
    \qquad
    \infovec_{k|k-1} := \infomat_{k|k-1}\stateest_{k|k-1}.
\end{equation}
}%
Since the global state $\state_k=\col(\state_{1,k},\dots,\state_{\Nrob,k})$ is stacked agent-wise, the matrix $\infomat_{k|k-1}\in\mathbb R^{\Nrob d\times \Nrob d}$ is naturally partitioned into $\Nrob\times \Nrob$ blocks, where the block $\infomat_{ij,k|k-1}\in\mathbb R^{d\times d}$ captures the coupling between the information corresponding to the states of agents $i$ and $j$.

\begin{assumption}
\label{ass:sparse_init}
The initial information matrix $\infomat_{0|0}$ is symmetric positive definite and has sparsity pattern consistent with the communication graph $\graph_c=(\vertices,\edges_c)$, in the sense that
{\small
\[
\infomat_{ij,0|0}=0,
\qquad
\forall\, i\neq j \text{ with } \{i,j\}\notin \edges_c,
\]
}%
where $\infomat_{ij,0|0}$ denotes the $(i,j)$ block of $\infomat_{0|0}$. Hence, nonzero off-diagonal blocks are allowed only between neighboring agents.
\end{assumption}

\subsection{Centralized information update in additive form}
\label{subsec:central_info}

Let $\boldsymbol{E}_i\in\mathbb R^{\Nrob d\times d}$ denote the block selection matrix such that
$\boldsymbol{E}_i^\top \state = \state_i$ and $\state = \sum_{i=1}^{\Nrob} \boldsymbol{E}_i \state_i$.
For an edge $(i,j)\in\edges_s$, define $\boldsymbol{E}_{ij}:=[\boldsymbol{E}_i\;\;\boldsymbol{E}_j]\in\mathbb R^{\Nrob d\times 2d}$.

For each local measurement $\meas^{\ell}_{i,k} = \measmat^{\ell}_{i,k}\state_{i,k}$, define the corresponding global measurement matrix $\measmat^{\ell}_{i,k}\boldsymbol{E}_i^\top$. For each relative measurement $\meas^r_{ij,k} = \measmat_{ij,i,k}^{r}\state_{i,k}+\measmat_{ij,j,k}^{r}\state_{j,k}$, define
{\small
\[
\measmat^r_{ij,k} :=
\begin{bmatrix}
\measmat_{ij,i,k}^{r} & \measmat_{ij,j,k}^{r}
\end{bmatrix}
\in\mathbb R^{q_r\times 2d},
\qquad
\measmat^r_{ij,k}\boldsymbol{E}_{ij}^\top \state
=
\meas^r_{ij,k}.
\]
}%

We express the measurement update by explicitly separating local and relative measurements, as in Section~II. Let $\vertices_{\ell} \subseteq \vertices$ be the set of agents equipped with local measurements \eqref{eq:private_meas}. For each $(i,j)\in\edges_s$, let the relative measurement be given by \eqref{eq:relative_meas}. Assume positive definite measurement information gains $\outcov^{\ell}_{i,k}\succeq \underline{r}^{\ell}\eye{q_\ell}$ for $i\in\vertices_{\ell}$, $\outcov^r_{ij,k}\succeq \underline{r}^r\eye{q_r}$ for $(i,j)\in\edges_s$, and
{\small
\[
\outcov_k = \blkdiag\!\bigl(\{\outcov^{\ell}_{i,k}\}_{i \in \vertices_{\ell}}, \{\outcov^r_{ij,k}\}_{(i,j) \in \edges_s}\bigr),
\]
}%
with bounded variation $\norm{\outcov_{k+1} - \outcov_k}\le \bar{r}_{\Delta}$.

Then the centralized information update can be written as
{\small
\begin{subequations}\label{eq:info_additive}
\begin{align}
\infomat_{k|k}
&=
\infomat_{k|k-1}
+
\smallgain
\measmat_k^\top \outcov_k
\measmat_k \nonumber
\\
&=
\infomat_{k|k-1}
+
\smallgain
\sum_{i\in\vertices_{\ell}}
\boldsymbol{E}_i \big(\measmat^{\ell}_{i,k}\big)^\top 
\outcov^{\ell}_{i,k}
\measmat^{\ell}_{i,k} \boldsymbol{E}_i^\top
\nonumber\\
&\hspace{12mm}+
\smallgain
\sum_{(i,j)\in\edges_s}
\boldsymbol{E}_{ij}\big(\measmat^r_{ij,k}\big)^\top 
\outcov^r_{ij,k}
\measmat^r_{ij,k} \boldsymbol{E}_{ij}^\top,
\label{eq:info_additive_S}
\\
\infovec_{k|k}
&=
\infovec_{k|k-1}
+
\smallgain
\sum_{i\in\vertices_{\ell}}
\boldsymbol{E}_i \big(\measmat^{\ell}_{i,k}\big)^\top 
\outcov^{\ell}_{i,k}
\meas^{\ell}_{i,k}
\nonumber\\
&\hspace{12mm}+
\smallgain
\sum_{(i,j)\in\edges_s}
\boldsymbol{E}_{ij}\big(\measmat^r_{ij,k}\big)^\top 
\outcov^r_{ij,k}
\meas^r_{ij,k}
\label{eq:info_additive_z}
\\
\stateest_{k|k} &= \infomat_{k|k}^{-1}\infovec_{k|k}.
\label{eq:info_inverse_mapping}
\end{align}
\end{subequations}
}%
Here, $\smallgain\in(0,1]$ is a correction gain that will later be used to tune the observer while preserving the additive structure of the update. For $\smallgain=1$, \eqref{eq:info_additive} reduces to the standard information-form Kalman update.

Equation \eqref{eq:info_additive_S} also makes explicit that the measurement update preserves the graph-induced sparsity pattern. This observation is formalized next.

\begin{proposition}[Sparsity preservation of the measurement update]
\label{prop:update_sparsity}
If $\infomat_{k|k-1}$ has sparsity pattern consistent with the communication graph $\graph_c$, then the additive measurement update \eqref{eq:info_additive_S} preserves the same sparsity pattern.
\end{proposition}

\begin{proof}
From \eqref{eq:info_additive_S}, each local contribution is premultiplied and postmultiplied by $\boldsymbol{E}_i$, and therefore affects only the diagonal block associated with agent $i$. Likewise, each relative contribution associated with an edge $(i,j)\in\edges_s$ is premultiplied and postmultiplied by $\boldsymbol{E}_{ij}$ and $\boldsymbol{E}_{ij}^\top$, and therefore affects only the $2\times 2$ block submatrix indexed by agents $i$ and $j$. Hence, no nonzero blocks are created outside the graph-induced sparsity pattern.
\end{proof}

The additive structure in \eqref{eq:info_additive} has been widely used in sensor networks estimating a \emph{global} state of fixed dimension (e.g., target tracking) \cite{battistelli2014consensus, kamal2013information,sebastian2023eco}, where the sums can be computed via average consensus.
However, in cooperative localization the state dimension grows with $\Nrob$ and each agent ultimately requires only its \emph{local} estimate.
Thus, neither centralized inversion of $\infomat_{k|k}$ (as required by \eqref{eq:info_inverse_mapping}) nor consensus over global quantities is scalable.

\subsection{Sparsity-preserving prediction}
\label{subsec:sparsity_pred}

A key difficulty in applying information filtering to the considered setup is that the standard prediction step
{\small
\[
    \infomat_{k+1|k} = \big(\statemat_k\covmat_{k|k}\statemat_k^\top + \procnoisemat_k\big)^{-1}
\]
}%
typically becomes dense due to the additive process noise covariance $\procnoisemat_k$.
To preserve sparsity, similarly to~\cite{de2025distributed}, we adopt a modified prediction step based on exponential forgetting, namely
{\small
\begin{equation}
    \infomat_{k+1|k}
    =
    \ffactor\,
    \statemat_k^{-\top}\infomat_{k|k}\statemat_k^{-1},
    \qquad 0<\ffactor<1.
    \label{eq:forgetting_prediction_clean}
\end{equation}
}%
Since $\statemat_k$ is block diagonal, \eqref{eq:forgetting_prediction_clean} preserves the sparsity pattern of $\infomat_{k|k}$.

\begin{proposition}[Sparsity preservation of the information prediction]
\label{prop:sparsity}
For the multi-agent system \eqref{eq:local_dyn}--\eqref{eq:global_meas}, assume that $\infomat_{k|k}$ has sparsity pattern consistent with the communication graph $\graph_c$. Then the modified prediction step \eqref{eq:forgetting_prediction_clean} preserves the same sparsity pattern.
\end{proposition}

\begin{proof}
In the considered setup, the dynamics matrix $\statemat_k$ 
is block diagonal. Hence, $\statemat_k^{-1}$ is block diagonal as well. Therefore, the $(i,j)$ block of $\infomat_{k+1|k}$ is
{\small
\[
\infomat_{ij,k+1|k}
=
\ffactor
\statemat_{i,k}^{-\top}
\infomat_{ij,k|k}
\statemat_{j,k}^{-1}
.
\]
}%
It follows that $\infomat_{ij,k+1|k}=0$ whenever
$\infomat_{ij,k|k}=0$. Hence, the sparsity pattern is preserved.
\end{proof}

A diagonal matrix forgetting factor can also be considered, namely
{\small
\[
\begin{aligned}
\ffmatrix&=\blkdiag(\ffmatrix_i)_{i\in\vertices},
\qquad
\ffmatrix_i=\diag(\ffactor_{i,n})_{n=1}^d,
\\
&0<\ffactor_{i,n}<1,\qquad \forall n\in\{1,\dots,d\},
\end{aligned}
\]
}%
together with the prediction
{\small
\begin{equation}
    \infomat_{k+1|k}
    =
    \statemat_k^{-\top}\ffmatrix\,\infomat_{k|k}\,\ffmatrix\,\statemat_k^{-1}.
    \label{eq:forgetting_prediction_matrix}
\end{equation}
}%
This variant may be desirable in practice, as it allows different rates of forgetting to be tuned differently across agents and/or state components while still preserving sparsity. Since the corresponding stability conditions are more involved and generally conservative, the stability of the matrix-forgetting case is discussed separately in Appendix~\ref{app:matrixffuges} and here we focus on the scalar case \eqref{eq:forgetting_prediction_clean}.

By Assumption~\ref{ass:sparse_init} together with Propositions~\ref{prop:update_sparsity} and~\ref{prop:sparsity}, the information matrix retains a sparsity pattern consistent with $\graph_c$ for all $k$.

\begin{lemma}
\label{lemma:conditioning}
    Under Assumptions~\ref{ass:boundedness}--\ref{ass:invertibility} and choosing $\ffactor\le 1 / a_{\mathrm{inv}}^2$ 
    (see \cite[Lemma~1]{cticlea2013exponential}), the matrix $\infomat_{k|k}$ is such that
    {\small
    \begin{equation}
        \zeros{\Nrob d \times \Nrob d} \prec \underline{s}_\smallgain
        \eye{\Nrob d} \preceq \infomat_{k|k} \preceq \bar{s}_\smallgain
        \eye{\Nrob d}
    \end{equation}
    }
    with $\underline{s}_\smallgain = \smallgain \underline{s}$, $\bar{s}_\smallgain = \smallgain\bar{s}$ and $\underline{s},\bar{s}>0$.
\end{lemma}
\begin{proof}
    Denote by $\bar{\infomat}_{k|k}$ the information matrix generated by the same recursion as $\infomat_{k|k}$ with $\smallgain=1$ and initialization $\bar{\infomat}_{0|0}$ such that it holds $\infomat_{0|0}
    = \smallgain \bar{\infomat}_{0|0}$. Then, using recursively the prediction step \eqref{eq:forgetting_prediction_clean} and the additive update
    \eqref{eq:info_additive_S}, one obtains
    {\small
    \begin{equation}
    \label{eq:S_expansion}
    \begin{aligned}
    \infomat_{k|k}
    =
    \smallgain&\!\left[
    \ffactor^k\,\boldsymbol{\Phi}(0,k)^\top \bar{\infomat}_{0|0}\boldsymbol{\Phi}(0,k)
    \right.
    \\
    &\left.+
    \sum_{h=0}^{k}
    \ffactor^{k-h}\,
    \boldsymbol{\Phi}(h,k)^\top
    \measmat_h^\top \outcov_h \measmat_h
    \boldsymbol{\Phi}(h,k)
    \right]
    \end{aligned}
    \end{equation}
    }%
    This shows that $\infomat_{k|k} = \smallgain \bar{\infomat}_{k|k}$, i.e. the resulting information matrix is a scaled version of the one obtained when $\smallgain=1$. By \cite[Lemma~1]{cticlea2013exponential}, there exist constants $\underline{s},\bar{s}>0$ such that
    {\small
    \[
        \zeros{} < \underline{s}\,\eye{\Nrob d}
        \leq \bar{\infomat}_{k|k} \leq \bar{s}\,\eye{\Nrob d}.
    \]
    }%
    Multiplying by $\smallgain$ yields the stated result.
\end{proof}

Hence, under Assumption~\ref{ass:sparse_init}, the correction step and the modified prediction preserve the graph-induced sparsity pattern of $\infomat_{k|k}$, while Lemma~\ref{lemma:conditioning} guarantees that $\infomat_{k|k}$ remains symmetric positive definite. Consequently, as it will be shown in the next section, the posterior update can be recast as the solution of a sparse symmetric positive definite linear system, which is the starting point for the distributed observer developed next.


\subsection{Correction form and local observer dynamics}
\label{subsec:correction_linear}

Define the correction $\correction_k := \frac{1}{\smallgain}(\stateest_{k|k}-\stateest_{k|k-1})$.
The Kalman update can be written in correction form as
{\small
\begin{equation}
    \stateest_{k|k}
    =
    \stateest_{k|k-1}
    +
    \smallgain
    \correction_k,
    \qquad
    \infomat_{k|k}\correction_k = \knownterm_k,
    \label{eq:correction_linear_clean}
\end{equation}
}%
where the right-hand side $\knownterm_k$ collects all innovation terms:
{\small
\begin{equation}
\label{eq:b_def_clean}
\begin{aligned}
    \knownterm_k &:=
    \measmat_k^\top \outcov_k
    (\meas_k - \measmat_{k}\stateest_{k|k-1}).
\end{aligned}
\end{equation}
}%
Since this expression is not the common form of the Kalman update, we report its derivation in Appendix~\ref{app:kalupdate} for the reader's convenience.

Equation \eqref{eq:correction_linear_clean} shows that computing the posterior estimate is equivalent to solving a sparse symmetric positive definite linear system. While this remains equivalent to inverting $\infomat_{k|k}$, the key advantage is that the sparsity of the system enables a distributed iterative solution, rather than requiring a centralized matrix inversion.

\begin{remark}[On the role of $\smallgain$]
The gain $\smallgain\in(0,1]$ does not change the nominal Kalman-like observer dynamics. In fact, since $\infomat_{k|k}=\smallgain\,\bar{\infomat}_{k|k}$ as per Lemma~\ref{lemma:conditioning}, the correction equation
\[
\infomat_{k|k}\correction_k=\knownterm_k
\]
implies $\correction_k=\smallgain^{-1}\bar{\infomat}_{k|k}^{-1}\knownterm_k$, so that the update term actually applied to the estimate is
\[
\smallgain\,\correction_k=\bar{\infomat}_{k|k}^{-1}\knownterm_k,
\]
which is the same as for $\smallgain=1$. Hence, $\smallgain$ should be interpreted not as a modification of the nominal observer correction, but rather as a parameter that regulates the interaction with the distributed correction dynamics.
\end{remark}

Before introducing the distributed algorithm used to solve \eqref{eq:correction_linear_clean}, we make the agent-wise structure of the observer explicit. 
To this end, define the indicators 
{\small\begin{equation}\label{eq:meas_indicators}
\delta_i^{\ell} \define
\begin{cases}
1, & i\in\vertices_{\ell},\\
0, & \text{otherwise},
\end{cases}
\qquad
\delta_{ij}^r \define
\begin{cases}
1, & (i,j)\in\edges_s,\\
0, & \text{otherwise}.
\end{cases}
\end{equation}}%
Then the $i$-th block of $\knownterm_k$ depends only on local quantities, namely
{\small
\begin{equation}
\label{eq:b_local_block}
\begin{aligned}
    \knownterm_{i,k}
    &=
    \delta_i^{\ell}
    (\measmat^{\ell}_{i,k})^\top 
    \outcov^{\ell}_{i,k}
    \bigl(\meas^{\ell}_{i,k} - \measmat^{\ell}_{i,k}\stateest_{i,k|k-1}\bigr)
    \\
    &\quad+
    \sum_{j \in \neigh_i}
    \delta_{ij}^r
    (\measmat_{ij,i,k}^{r})^\top 
    \outcov^r_{ij,k}
    \bigl(\meas^r_{ij,k} - \measmat^r_{ij,k}\boldsymbol{E}_{ij}^\top\stateest_{k|k-1}\bigr)
    \\
    &\quad+
    \sum_{j \in \neigh_i}
    \delta_{ji}^r
    (\measmat_{ji,i,k}^{r})^\top 
    \outcov^r_{ji,k}
    \bigl(\meas^r_{ji,k} - \measmat^r_{ji,k}\boldsymbol{E}_{ji}^\top\stateest_{k|k-1}\bigr).
\end{aligned}
\end{equation}
}%

We next make explicit the observer dynamics from the perspective of agent $i$. Define the local and relative information increments as
{\small
\begin{equation*}
\begin{aligned}
    \Delta\infomat^{\ell}_{i,k}
    &\define
    \smallgain\delta_i^{\ell}
    (\measmat^{\ell}_{i,k})^\top 
    \outcov^{\ell}_{i,k}
    \measmat^{\ell}_{i,k},
    \\
    \Delta\infomat^{r}_{ij,ii,k}
    &\define
    \smallgain\delta^r_{ij}
    (\measmat_{ij,i,k}^{r})^\top 
    \outcov^r_{ij,k}
    \measmat_{ij,i,k}^{r},
    \\
    \Delta\infomat^{r}_{ji,ii,k}
    &\define
    \smallgain\delta^r_{ji}
    (\measmat_{ji,i,k}^{r})^\top 
    \outcov^r_{ji,k}
    \measmat_{ji,i,k}^{r},
    \\
    \Delta\infomat^{r}_{ij,ij,k}
    &\define
    \smallgain\delta^r_{ij}
    (\measmat_{ij,i,k}^{r})^\top 
    \outcov^r_{ij,k}
    \measmat_{ij,j,k}^{r},
    \\
    \Delta\infomat^{r}_{ji,ij,k}
    &\define
    \smallgain\delta^r_{ji}
    (\measmat_{ji,i,k}^{r})^\top 
    \outcov^r_{ji,k}
    \measmat_{ji,j,k}^{r} .
\end{aligned}
\end{equation*}
}%
Then, from \eqref{eq:correction_linear_clean} and \eqref{eq:info_additive_S}, the local update can be written as
{\small
\begin{equation}
\label{eq:local_update}
\begin{aligned}
    \stateest_{i,k|k}
    &=
    \stateest_{i,k|k-1}
    +
    \smallgain
    \correction_{i,k},
    \\
    \infomat_{ii,k|k}
    &=
    \infomat_{ii,k|k-1}
    +
    \Delta\infomat^{\ell}_{i,k}
    +
    \sum_{j\in\neigh_i}
    \bigl(
    \Delta\infomat^{r}_{ij,ii,k}
    +
    \Delta\infomat^{r}_{ji,ii,k}
    \bigr)
    ,
    \\
    \infomat_{ij,k|k}
    &=
    \infomat_{ij,k|k-1}
    +
    \Delta\infomat^{r}_{ij,ij,k}
    +
    \Delta\infomat^{r}_{ji,ij,k}
    ,
    \qquad \forall j\in\neigh_i.
\end{aligned}
\end{equation}
}%

The local prediction step can be written as
{\small
\begin{equation}
\label{eq:local_prediction}
\begin{aligned}
    \stateest_{i,k+1|k}
    &=
    \statemat_{i,k}\stateest_{i,k|k}
    +
    \inputmat_{i,k}\sysinput_{i,k},
    \\
    \infomat_{ij,k+1|k}
    &=
    \ffactor
    \statemat_{i,k}^{-\top}
    \infomat_{ij,k|k}
    \statemat_{j,k}^{-1}
    \qquad \forall j \in \{i\}\cup\neigh_i.
\end{aligned}
\end{equation}
}%




\subsection{Quadratic program and separable structure}
\label{subsec:qp_structure}

Equation~\eqref{eq:correction_linear_clean} shows that computing the correction amounts to solving a sparse symmetric positive definite linear system. Rather than explicitly forming $\infomat_{k|k}^{-1}$, we seek $\correction_k$ through a distributed iterative scheme.

\begin{remark}[Richardson iteration]
As it was shown in \cite{de2025distributed} in a continuous-time setting, a solution of \eqref{eq:correction_linear_clean} can be obtained in a distributed way using the Richardson iteration, which we write here in a discrete-time setting as
{\small
\begin{equation}
    \estcorrection_{k+1}
    =
    \estcorrection_k
    -
    \alpha_{R,k}\bigl(\infomat_{k|k}\estcorrection_k-\knownterm_k\bigr),
    \qquad
    0<\alpha_{R,k}<\frac{2}{\lambda_{\max}(\infomat_{k|k})}.
    \label{eq:richardson_clean}
\end{equation}
}%
Since $\infomat_{k|k}$ is sparse and positive definite, \eqref{eq:richardson_clean} can be implemented in a distributed way and converges for fixed $\knownterm_k$. However, its convergence may deteriorate when $\infomat_{k|k}$ is ill-conditioned, and selecting the step size is not straightforward, because reliable upper bounds on $\lambda_{\max}(\infomat_{k|k})$ are typically unavailable a priori and may vary substantially, e.g., with the size of the network. This motivates the ADMM-based distributed correction proposed next.
\end{remark}

By Lemma~\ref{lemma:conditioning}, the information matrix satisfies $\infomat_{k|k}\succeq \underline{s}_\smallgain \eye{\Nrob d}$ with $\underline{s}_\smallgain>0$. Hence, the linear system \eqref{eq:correction_linear_clean} admits a unique solution and can be equivalently characterized as the first-order optimality condition of the strongly convex quadratic program
{\small
\begin{equation}
\label{eq:qp_clean}
    \correction_k
    =
    \arg\min_{\correction\in\mathbb R^{\Nrob d}}
    \left\{
        \frac12\,\correction^\top \infomat_{k|k}\correction
        -
        \correction^\top \knownterm_k
    \right\}.
\end{equation}
}%
This reformulation is natural for two reasons. First, the sparsity of $\infomat_{k|k}$ induces a decomposition of the objective into local and pairwise coupling terms. Second, \eqref{eq:qp_clean} is equivalent, up to an additive constant, to minimizing the weighted squared residual of \eqref{eq:correction_linear_clean}, namely
{\small
\[
\frac12\,\norm{\infomat_{k|k}\correction-\knownterm_k}_{\infomat_{k|k}^{-1}}^2.
\]
}%
The weighting by $\infomat_{k|k}^{-1}$ is crucial. Let $J$ denote the objective in \eqref{eq:qp_clean}. Then
{\small
\[
\nabla J(\correction)=\infomat_{k|k}\correction-\knownterm_k,
\]
}%
where we denoted by $J$ the cost in \eqref{eq:qp_clean} and whose $i$-th block depends only on $\correction_i$ and on the corrections of the neighbors of agent $i$. By contrast, minimizing the unweighted residual norm
{\small 
\[
\norm{\infomat_{k|k}\correction-\knownterm_k}_2^2
\]
}%
would produce a gradient involving $\infomat_{k|k}^2$, which generally introduces couplings over two-hop neighborhoods. Moreover, the Richardson iteration \eqref{eq:richardson_clean} can be interpreted as gradient descent applied to \eqref{eq:qp_clean}.

We now exploit the sparsity structure of $\infomat_{k|k}$ to express \eqref{eq:qp_clean} as a sum of local and edge terms.
To obtain a distributed formulation yielding positive-semidefinite local costs, we exploit the additive structure of the information and keep separate the information contributions due to local measurements and those due to relative measurements. 
Specifically, for each agent $i\in\vertices$, we introduce the local information contribution $\infomat_{i,k|k}^{\ell}\in\nR{d\times d}$ and the corresponding innovation term $\knownterm_{i,k}^{\ell}\in\nR{d}$, updated as
{\small
\begin{equation}\label{eq:loc_info_update}
\begin{aligned}
\infomat_{i,k|k}^{\ell}
&\define
\infomat_{i,k|k-1}^{\ell} +
\Delta\infomat^{\ell}_{i,k},
\\
\knownterm_{i,k}^{\ell}
&\define
\delta_i^\ell
 (\measmat_{i,k}^{\ell})^\top
\outcov_{i,k}^{\ell}
\bigl(
\meas_{i,k}^{\ell}-\measmat_{i,k}^{\ell}\stateest_{i,k|k-1}
\bigr).
\end{aligned}
\end{equation}
}%

Likewise, for each undirected edge $\{i,j\}\in\edges_c$, we introduce the pairwise information contribution 
{\small 
\[
\infomat_{ij,k|k}^{e} \define \begin{bmatrix}
   \infomat_{ij,ii,k|k}^e
   & \infomat_{ij,ij,k|k}^e
   \\
   \infomat_{ij,ji,k|k}^e
   & \infomat_{ij,jj,k|k}^e
\end{bmatrix}  \in\nR{2d\times 2d}
\]
}%
and the corresponding innovation term 
{\small
\[
\knownterm_{ij,k}^{e} \define
\begin{bmatrix}
    \knownterm_{ij,i,k} \\ \knownterm_{ij,j,k}
\end{bmatrix}
\in\nR{2d},
\]
}%
computed as
{\small
\begin{equation}\label{eq:rel_info_update}
\begin{aligned}
    \infomat_{ij,k|k}^{e}
&=
\infomat_{ij,k|k-1}^{e}
\\
&+
\begin{bmatrix}
\Delta\infomat^{r}_{ij,ii,k} + \Delta\infomat^{r}_{ji,ii,k} & 
\Delta\infomat^{r}_{ij,ij,k} + \Delta\infomat^{r}_{ji,ij,k}
\\
\Delta\infomat^{ij}_{ji,k} + \Delta\infomat^{ji}_{ji,k}
&
\Delta\infomat^{ij}_{jj,k} + \Delta\infomat^{ji}_{jj,k}
\end{bmatrix}
\\
\knownterm_{ij,k}^{e}
&=
\delta_{ij}^r
    \begin{bmatrix}
        (\measmat_{ij,i,k}^{r})^\top \\ (\measmat_{ij,j,k}^{r})^\top
    \end{bmatrix}
    \outcov^r_{ij,k}
    \bigl(\meas^r_{ij,k} - \measmat^r_{ij,k}\boldsymbol{E}_{ij}^\top\stateest_{k|k-1}\bigr)
\\
&+
\delta_{ji}^r
    \begin{bmatrix}
        (\measmat_{ji,i,k}^{r})^\top \\ (\measmat_{ji,j,k}^{r})^\top
    \end{bmatrix} 
    \outcov^r_{ji,k}
    \bigl(\meas^r_{ji,k} - \measmat^r_{ji,k}\boldsymbol{E}_{ji}^\top\stateest_{k|k-1}\bigr)    
\end{aligned}
\end{equation}
}%
where the indicator $\delta_{ij}^r$ was introduced in \eqref{eq:meas_indicators}.
Here, the pairwise term associated with the undirected edge $\{i,j\}$ collects the information contributions of both oriented relative measurements $(i,j)$ and $(j,i)$, whenever present. In particular, the indicators $\delta_{ij}^r$ and $\delta_{ji}^r$ allow either contribution to be present independently, so the construction does not require both directed edges to belong to $\mathcal E_s$. Moreover, $\infomat_{ij,k|k}^{e}$ and $\infomat_{ji,k|k}^{e}$ denote the pairwise information matrices associated with the same undirected edge, but written with opposite ordering of the state blocks. By construction, $\infomat_{i,k|k}^{\ell}\succeq 0$ and $\infomat_{ij,k|k}^{e}\succeq 0$.

The corresponding prediction step is obtained by applying the same sparsity-preserving forgetting-factor recursion introduced above to each of the information contributions $\infomat_{i,k|k}^{\ell}$ and $\infomat_{ij,k|k}^{e}$.

Define the quadratic cost associated to the local measurements as
{\small
\begin{equation}
\label{eq:private_qp_term_corrected}
    \phi_i(\correction_i)
    :=
    \frac12\,\correction_i^\top \infomat_{i,k|k}^{\ell}\correction_i
    -
    \correction_i^\top \knownterm_{i,k}^{\ell},
\end{equation}
}%
and for each edge $(i,j)\in\edges_c$, define the pairwise quadratic cost associated with relative measurements as
{\small
\begin{equation}
\label{eq:edge_qp_term_corrected}
    \phi_{ij}(\correction_i,\correction_j)
    :=
    \frac12
    \begin{bmatrix}
    \correction_i\\
    \correction_j
    \end{bmatrix}^\top
    \infomat_{ij,k|k}^{e}
    \begin{bmatrix}
    \correction_i\\
    \correction_j
    \end{bmatrix}
    -
    \begin{bmatrix}
    \correction_i\\
    \correction_j
    \end{bmatrix}^\top
    \knownterm_{ij,k}^{e}.
\end{equation}
}%

With these definitions, the global objective in \eqref{eq:qp_clean} can be written exactly as
{\small
\begin{equation}
\label{eq:qp_decomp_corrected}
\begin{aligned}
    J(\correction)
    &=
    \sum_{i=1}^{\Nrob}
        \left[
            \phi_i(\correction_i)
            + 
            \frac12 \sum_{j \in \neigh_i}
        \phi_{ij}(\correction_i,\correction_j)\right].
\end{aligned}
\end{equation}
}%
The first sum collects the contributions due to local measurements, while the second sum contains the full quadratic terms induced by relative measurements. The factor $\frac12$ compensates for the fact that each undirected edge appears twice in the neighbor sums. This decomposition is exact and preserves the positive-semidefinite structure of each individual information contribution. 

As a final remark, the above decomposition only changes the internal representation of the information variables, not the observer itself. Indeed, the global information matrix and innovation vector used in the original correction form are recovered by summing the local contributions and the pairwise edge contributions introduced above. Therefore, the state update and prediction steps in \eqref{eq:local_update} and \eqref{eq:local_prediction} are preserved, with the only modification that the information associated with local measurements and relative measurements is now stored and propagated separately.


\subsection{Distributed solution via partition-based ADMM}
\label{subsec:admm_details}

Problem \eqref{eq:qp_clean} is a strongly convex quadratic program
whose objective decomposes into local and pairwise coupling terms
as shown in \eqref{eq:qp_decomp_corrected}.  
Each agent $i$ is coupled only with its neighbors
$j\in\neigh_i$ through the off-diagonal blocks
$\infomat_{ij,k|k}$.
This structure enables a distributed solution using
the partition-based ADMM method~\cite{bastianello2018partition,shorinwa2020scalable}.

Each agent $i$ maintains a local copy of its own correction
$\correction_i$ and local copies of the corrections of its neighbors.
Specifically, agent $i$ stores
\[
\extloccorr{i}
=
\begin{bmatrix}
\localcorrection{i}{i}
\\
\col(\localcorrection{j}{i})_{j \in \neigh_i}
\end{bmatrix},
\]
where $\localcorrection{i}{i}$ denotes the local estimate of $\correction_i$
and $\localcorrection{j}{i}$ denotes agent $i$'s local copy of $\correction_j$.

Using the decomposition \eqref{eq:qp_decomp_corrected},
the local cost function at agent $i$ can be written as
{\small
\begin{equation}
\label{eq:local_cost}
\begin{aligned}
J_i(\extloccorr{i})
&=
\phi_i(\localcorrection{i}{i})
            + 
            \frac12 \sum_{j \in \neigh_i}
         \phi_{ij}(\localcorrection{i}{i},\localcorrection{j}{i}).
\end{aligned}
\end{equation}
}%
To ensure consistency among local copies, consensus is imposed by introducing equality constraints that require neighboring agents to agree on their shared variables, namely
{\small
\begin{equation}
\label{eq:consensus_constraints}
\localcorrection{i}{i} = \localcorrection{i}{j},
\qquad
\localcorrection{j}{i} = \localcorrection{j}{j},
\qquad
\forall \{i,j\}\in\edges_c .
\end{equation}
}%
With these consensus constraints, the optimization problem can be written in a form suitable for distributed computation as
{\small
\begin{equation}
\label{eq:admm_global_problem}
\begin{aligned}
\min_{\{\extloccorr{i}\}_{i=1}^{\Nrob}} \quad &
\sum_{i=1}^{\Nrob} J_i(\extloccorr{i})
&
\text{s.t.}\quad &
\eqref{eq:consensus_constraints}.
\end{aligned}
\end{equation}
}%
Problem~\eqref{eq:admm_global_problem} reveals that the correction step has the standard structure of a distributed optimization problem: the cost is separable across agents, and coupling is captured only by consistency constraints on shared variables. Consequently, once the problem is cast in the form \eqref{eq:admm_global_problem}, a variety of distributed optimization algorithms could in principle be employed. In this paper, we focus on ADMM, due to its favorable properties in terms of convergence and robustness.

Following~\cite[Prop.~1]{bastianello2018partition}, the relaxed ADMM iterations consist of two steps.

\textbf{1) Local primal update.}
Each agent solves the quadratic subproblem
{\small
\begin{equation}
\label{eq:admm_primal}
\begin{aligned}
    \estextloccorr{i}{k,\iteradmm}
&=
\argmin_{\extloccorr{i}}
\left\{
J_i(\extloccorr{i})
-
\sum_{j\in\neigh_i}
\big(
\boldsymbol{q}_{ij,i,k, \iteradmm}^\top \localcorrection{i}{i}
+
\boldsymbol{q}_{ij,j,k, \iteradmm}^\top \localcorrection{j}{i}
\big)
\right.
\\
&\left.+
\frac{\rho}{2}
\left(
|\neigh_i|\,\|\localcorrection{i}{i}\|^2
+
\sum_{j\in\neigh_i}\|\localcorrection{j}{i}\|^2
\right)
\right\}.
\end{aligned}
\end{equation}
}%

\textbf{2) Dual update.}
For every $\{i,j\}\in\edges_c$, the dual variables are updated as
{\small
\begin{equation}\label{eq:aux_vars_update}
\begin{aligned}
\boldsymbol{q}_{ij,i,k,\iteradmm+1}
&=
(1-\alpha)\boldsymbol{q}_{ij,i,k,\iteradmm}
+
\alpha\big(
-\boldsymbol{q}_{ji,i,k,\iteradmm}
+
2\rho\,\estlocalcorrection{i,k,\iteradmm}{j}
\big),
\\
\boldsymbol{q}_{ij,j,k,\iteradmm+1}
&=
(1-\alpha)\boldsymbol{q}_{ij,j,k,\iteradmm}
+
\alpha\big(
-\boldsymbol{q}_{ji,j,k,\iteradmm}
+
2\rho\,\estlocalcorrection{j,k,\iteradmm}{j}
\big),
\end{aligned}
\end{equation}
}%
where $\rho>0$ is the ADMM penalty parameter, $\alpha\in(0,1)$ is a relaxation factor, {and $\iteradmm$ denotes the ADMM iteration index. In particular, at each observer step we allow for the execution of multiple ADMM iterations, denoted by $\maxiteradmm\geq 1$. 
}

Because $J_i$ is quadratic and convex, and the ADMM penalty term makes \eqref{eq:admm_primal} strongly convex, the corresponding Hessian $\boldsymbol{H}_{\rho,i,k}$ is symmetric positive definite, and therefore invertible for all $i$ and $k$. As a result, the primal update admits the closed-form solution
{\small
\begin{equation}\label{eq:closed_form_primal}
\estextloccorr{i}{k,\iteradmm}
=
\boldsymbol{H}_{\rho,i,k}^{-1}\left( \bar{\knownterm}_{i,k} + \boldsymbol{A}_{qi}^\top \boldsymbol{q}_{i,k,\iteradmm} \right),
\end{equation}
}%
where
{\small
\begin{equation}\label{eq:def_qk_barbk}
\begin{aligned}
\bar{\knownterm}_{i,k}
&=
\begin{bmatrix}
    \knownterm_{i,k}^{\ell} + \frac12 \sum_{j \in \neigh_i} \knownterm_{ij,i,k}
    \\
    \col(\knownterm_{ij,j,k})_{j \in \neigh_i}
\end{bmatrix}
\\
\boldsymbol{q}_{i,k, \iteradmm}
&=
\begin{bmatrix}
\col(\boldsymbol{q}_{ij,i,k, \iteradmm})_{j\in\neigh_i}
\\
\col(\boldsymbol{q}_{ij,j,k,\iteradmm})_{j\in\neigh_i}
\end{bmatrix} .
\end{aligned}
\end{equation}
}%
Define {\small$\infomat_{\neigh_{i}, k|k}^{\times} \define \col(\infomat_{ij,ji,k|k}^{e})_{j \in \neigh_i}$}, {\small$\infomat_{\neigh_{i}, k|k}^{\diag} \define \blkdiag(\infomat_{ij,jj,k|k}^{e})_{j \in \neigh_i}$} and {\small$\infomat_{i,k|k}^{\mathrm{agg}} \define \infomat_{i,k|k}^{\ell} + 1/2 \sum_{j\in \neigh_i}\infomat_{ij,ii,k|k}^{e}$}, then the Hessian matrix of the quadratic cost \eqref{eq:admm_primal}, denoted $\boldsymbol{H}_{\rho,i,k}$, is given by
{\small
\begin{equation}\label{eq:H_decomp}%
\boldsymbol{H}_{\rho,i,k}
=
\rho \boldsymbol{H}_{i}^0 
+ 
\boldsymbol{H}_{i,k}^1,
\end{equation}
} 
with 
{
\small
\begin{equation*}
    \begin{aligned}
        \boldsymbol{H}_{i}^0 &= 
        \blkdiag(|\neigh_i|\eye{d}, \eye{|\neigh_i|d})
        \\
        \boldsymbol{H}_{i,k}^1
        &=
        \begin{bmatrix}
        \infomat_{i,k|k}^{\mathrm{agg}}
        &
        \frac{1}{2}
        (\infomat_{\neigh_{i}, k|k}^{\times})^\top
        \\[1mm]
        \frac{1}{2}
        \infomat_{\neigh_{i}, k|k}^{\times}
        &
        \frac12 \infomat_{\neigh_{i}, k|k}^{\diag}
        \end{bmatrix},
    \end{aligned}
\end{equation*}
}%
where it is important to note that, by Lemma~\ref{lemma:conditioning}, we can write $\boldsymbol{H}_{i,k}^1 = \smallgain \bar{\boldsymbol{H}}_{i,k}^1$, with $\bar{\boldsymbol{H}}_{i,k}^1$ independent of $\smallgain$. 
Finally, $\boldsymbol{A}_{qi}^\top$ is the matrix selecting the dual terms in the linear part of \eqref{eq:admm_primal}, namely
{\small
\begin{equation}\label{eq:defAqi}
\boldsymbol{A}_{qi}
=
\begin{bmatrix}
\ones{|\neigh_i|}\otimes \eye{d} &  \zeros{|\neigh_i|\cdot d \times |\neigh_i|\cdot d}
\\[1mm]
\zeros{|\neigh_i|\cdot d \times d} & \eye{|\neigh_i|\cdot d} 
\end{bmatrix},
\end{equation}
}%
where $\boldsymbol{A}_{qi}$ is such that
{\small
\[
\boldsymbol{A}_{qi}^\top \boldsymbol{q}_{i,k,\iteradmm}
=
\begin{bmatrix}
\sum_{j\in\neigh_i}\boldsymbol{q}_{ij,i,k,\iteradmm}
\\[1mm]
\verticalstack{\boldsymbol{q}}{ij,j,k,\iteradmm}{j}{\neigh_i}
\end{bmatrix}.
\]
}%

The complete distributed observer is summarized in Algorithm~\ref{alg:dist_observer_short}.

\begin{remark}[On graph connectivity and observability]
We stress that connectivity of neither the sensing graph nor the communication graph is required. The standing assumption is instead complete uniform observability, as stated in Assumption~\ref{ass:compobs}. In particular, the network may consist of multiple disconnected clusters. In this case, both the global observability Gramian and the information matrix have a block-diagonal structure, with one block associated with each connected component (cluster). Consequently, the optimization problem \eqref{eq:admm_global_problem} decouples into independent subproblems, one for each cluster. Assumption~\ref{ass:compobs} must then hold for each corresponding block. Hence, graph connectivity is not needed per se, provided that each disconnected cluster is uniformly observable through its own local and relative measurements.
\end{remark}

\begin{algorithm}[t]
\caption{Distributed Kalman-like observer}
\label{alg:dist_observer_short}
\begin{algorithmic}[1]
\Require {\small$\stateest_{i,0|0}$}, {\small$\infomat_{i,0|0}^{\ell}$}, {\small$\infomat_{ij,0|0}^{r}$} {\small$\forall j \in \neigh_i$}, $\ffmatrix$, $\rho$, $\alpha$
\For{$k=0,1,2,\dots$}
    \ForAll{$i\in\vertices$ \textbf{in parallel}}
        \State Acquire local and relative measurements
        \State Update {\small$\infomat_{i,k|k}^{\ell}$}, {\small$\knownterm_{i,k}^{\ell}$}, and {\small$\infomat_{ij,k|k}^{e}$}, {\small$\knownterm_{ij,k}^{e}$} for all $j\in\neigh_i$ according to \eqref{eq:loc_info_update} and \eqref{eq:rel_info_update}
        \State Set {\small$\boldsymbol{q}_{ij,i,k,0} = \boldsymbol{q}_{ij,i,k-1,\maxiteradmm}$} and {\small$\boldsymbol{q}_{ij,j,k,0} = \boldsymbol{q}_{ij,j,k-1,\maxiteradmm}$}, {\small$\forall j \in \neigh_i$}
        \For{{$\iteradmm=0,1,\dots, \maxiteradmm$}}
        \State Form the local correction subproblem and compute $\estextloccorr{i}{k,\iteradmm}$ as in \eqref{eq:closed_form_primal}
        \State Compute
        {\footnotesize
        \begin{equation*}
        \begin{aligned}
            \boldsymbol{\eta}_{i}^{i\to j}
            &\define
            -\boldsymbol{q}_{ij,i,k,\iteradmm}+2\rho\,\estlocalcorrection{i,k,\iteradmm}{i},
            \\
            \boldsymbol{\eta}_{j}^{i\to j}
            &\define
            -\boldsymbol{q}_{ij,j,k,\iteradmm}+2\rho\,\estlocalcorrection{j,k,\iteradmm}{i},
        \end{aligned}
        \end{equation*}}
        \State transmit {\small$\boldsymbol{\eta}_{i}^{i\to j}$} and {\small$\boldsymbol{\eta}_{j}^{i\to j}$} to each {\small$j\in\neigh_i$}
        \State Update dual variables using \eqref{eq:aux_vars_update}
        \EndFor
        \State Set {\small$\hat{\correction}_{i,k} =\estlocalcorrection{i,k,\maxiteradmm}{i}$} obtained from the primal step
        \State \textbf{Observer state update} as in \eqref{eq:local_update} using {\small$\hat{\correction}_{i,k}$}
        \State \textbf{Information prediction} of the local and pairwise information contributions analogous to \eqref{eq:local_prediction}
        \State \textbf{State prediction} as in \eqref{eq:local_prediction}
    \EndFor
\EndFor
\end{algorithmic}
\end{algorithm}


\subsection{Direct integration of local measurements}
\label{subsec:private_split}

As shown in Sec.~\ref{subsec:qp_structure}, the cost in \eqref{eq:qp_decomp_corrected} naturally decomposes into a component induced by local measurements and a component induced by relative measurements. This structure can be exploited to directly integrate the information provided by local measurements, which depend only on agent \(i\), without being affected by the additional dynamics associated with reaching consensus with neighboring agents.

Specifically, let
{\small
\[
\infomat_{k|k}=\infomat_{k|k}^\ell+\infomat_{k|k}^r,
\qquad
\knownterm_k=\knownterm_k^\ell+\knownterm_k^r,
\]
}%
where \(\infomat_{k|k}^\ell\) is block diagonal and collects the contributions due only to local measurements, while \(\infomat_{k|k}^r\) contains the contributions due to relative measurements. Then the correction equation
{\small
$
\infomat_{k|k}\correction_k=\knownterm_k
$
}%
can be written as
{\small
\begin{equation}
\label{eq:split_linear_system}
(\infomat_{k|k}^\ell+\infomat_{k|k}^r)\correction_k
=
\knownterm_k^\ell+\knownterm_k^r.
\end{equation}
}%
Equivalently, the corresponding quadratic cost can be decomposed as
{\small
\[
J(\correction)=J_\ell(\correction)+J_r(\correction),
\]
}%
where
{\small
\[
J_\ell(\correction)
=
\frac12\,\correction^\top \infomat_{k|k}^\ell \correction
-
\correction^\top \knownterm_k^\ell,
\qquad
J_r(\correction)
=
\frac12\,\correction^\top \infomat_{k|k}^r \correction
-
\correction^\top \knownterm_k^r.
\]
}%

Since \(\infomat_{k|k}^\ell\) is block diagonal, the minimization of \(J_\ell\) is fully decoupled across agents. Therefore, local measurements can be integrated immediately by solving the local problem
{\small
\begin{equation}
\label{eq:private_only_correction}
\correction_k^\ell
=
(\infomat_{k|k}^\ell)^{-1}\knownterm_k^\ell,
\end{equation}
}%
or, equivalently, agent-wise,
{\small
\[
\correction_{i,k}^\ell
=
(\infomat_{i,k|k}^{\ell})^{-1}\knownterm_{i,k}^\ell,
\qquad i\in\vertices.
\]
}%
This yields an instantaneous correction based only on local measurements, without requiring communication.

To account for the relative measurements without double counting the local information, one may then introduce the residual variable
{\small
\[
\correction_k=\correction_k^\ell+\boldsymbol{\eta}_k.
\]
}%
Substituting this expression into \eqref{eq:split_linear_system} gives
{\small
\[
(\infomat_{k|k}^\ell+\infomat_{k|k}^r)(\correction_k^\ell+\boldsymbol{\eta}_k)
=
\knownterm_k^\ell+\knownterm_k^r.
\]
}%
Using the identity
{\small
$
\infomat_{k|k}^\ell\correction_k^\ell=\knownterm_k^\ell,
$
}%
one obtains
{\small
\begin{equation}
\label{eq:residual_correction}
(\infomat_{k|k}^\ell+\infomat_{k|k}^r)\boldsymbol{\eta}_k
=
\underbrace{\knownterm_k^r-\infomat_{k|k}^r\correction_k^\ell}_{\define \knownterm_k^\eta}.
\end{equation}
}%
Hence, the residual correction solves \eqref{eq:residual_correction}
and the full correction is recovered as
{\small
\begin{equation}
\label{eq:full_split_correction}
\correction_k=\correction_k^\ell+\boldsymbol{\eta}_k.
\end{equation}
}%

Accordingly, when direct integration of local measurements is employed, the distributed optimization layer is no longer used to compute the full correction $\correction_k$, but rather the residual correction $\boldsymbol{\eta}_k$. The local term $\correction_k^\ell$ is computed independently by each agent using only its own measurements, whereas the distributed ADMM iterations are applied to \eqref{eq:residual_correction}. Once the residual correction has been computed, the full correction is recovered through \eqref{eq:full_split_correction}.

More precisely, the residual correction is the unique solution of the quadratic program
{\small
\begin{equation}
\label{eq:qp_residual}
\boldsymbol{\eta}_k
=
\argmin_{\boldsymbol{\eta}\in\mathbb{R}^{\Nrob d}}
\left\{
\frac12 \boldsymbol{\eta}^\top \infomat_{k|k}\boldsymbol{\eta}
-
\boldsymbol{\eta}^\top \knownterm_k^\eta
\right\},
\end{equation}
}%
where $\knownterm_k^\eta$ was defined in \eqref{eq:residual_correction}.
Equivalently, the distributed ADMM scheme developed in Sec.~\ref{subsec:admm_details} can be applied unchanged after replacing the global correction variable $\correction_k$ with the residual variable $\boldsymbol{\eta}_k$, and the linear term $\knownterm_k$ with $\knownterm_k^\eta$.

When direct integration of local measurements is used, Algorithm~\ref{alg:dist_observer_short} is modified as follows:
\begin{enumerate}
    \item after Line~4 of the algorithm, each agent first computes the local correction
    {\small
    \[
    \correction_{i,k}^{\ell}=(\infomat_{i,k|k}^{\ell})^{-1}\knownterm_{i,k}^{\ell};
    \]
    }%
    \item the distributed ADMM layer is then applied to the residual problem \eqref{eq:residual_correction}, using $\boldsymbol{\eta}_k$ as optimization variable, i.e.
    \begin{itemize}
        \item the primal step becomes: 
        {\small
        \begin{equation}
        \hat{\boldsymbol{\eta}}_{\neigh_{i},k,\iteradmm}
            =
            \boldsymbol{H}_{\rho,i,k}^{-1}\left( \bar{\knownterm}_{i,k}^\eta + \boldsymbol{A}_{qi}^\top \boldsymbol{q}_{i,k,\iteradmm} \right),
            \end{equation}
            }%
            where $\hat{\boldsymbol{\eta}}_{\neigh_{i},k,\iteradmm}$ and $\bar{\knownterm}_{i,k}^\eta$ are defined analogously to $\estextloccorr{i}{k,\iteradmm}$ and $\bar{\knownterm}_{i,k}$.
        \item the dual step is as in \eqref{eq:aux_vars_update} replacing $\estlocalcorrection{i,k,\iteradmm}{j}$ and $\estlocalcorrection{j,k,\iteradmm}{j}$ with $\hat{\boldsymbol{\eta}}_{i,k,\iteradmm}^{(j)}$ and $\hat{\boldsymbol{\eta}}_{j,k,\iteradmm}^{(j)}$
    \end{itemize}
    \item the full correction is reconstructed as
    {\small
    $
    \hat{\correction}_{i,k}=\correction_{i,k}^{\ell}+\hat{\boldsymbol{\eta}}_{i,k,\maxiteradmm}^{(i)}.
    $
    }%
\end{enumerate}
All subsequent observer and prediction updates are then performed using the correction $\hat{\correction}_{i,k}$.


For simplicity, in the following we focus on the algorithm introduced in the previous section. Nevertheless, the same analysis extends directly to the present decomposition.

\section{Stability Analysis}
\label{sec:stab_analysis}
This section establishes the stability of the proposed distributed observer. The analysis adopts a slow--fast viewpoint inspired by singular perturbation arguments \cite{carnevale2025admm, carnevale2025unifying}. After deriving the coupled error dynamics, we study separately the fast subsystem associated with the ADMM correction layer and the reduced estimation-error dynamics obtained for exact correction. These two ingredients are then combined through a comparison-based small-gain argument to establish uniform global exponential stability of the full interconnected system.

\subsection{Error dynamics}

Define {\small$\estcorrection_k \define \col(\hat{\correction}_{i,k})_{i \in \vertices}$} and the prior estimation error and correction error as
{\small
\begin{equation}
\label{eq:error_defs}
\estate_k \define \state_k-\stateest_{k|k-1},
\qquad
\ecorr_k \define \correction_k-\estcorrection_k .
\end{equation}
}%
Using the observer update
$\stateest_{k|k}
=
\stateest_{k|k-1}
+
\smallgain\estcorrection_{k}$,
together with the prediction
{\small
$\estate_{k+1}=\statemat_k(\state_k - \stateest_{k|k})$,
}%
and adding and subtracting the ideal correction term $\smallgain\correction_{k}$, the estimation error dynamics can be written as
%
{\small
\begin{equation}
\label{eq:err_dyn_full}
\begin{aligned}
\estate_{k+1}
&=
\statemat_k(\state_k-\stateest_{k|k})
\\
&=
\statemat_k\bigl(\state_k-\stateest_{k|k-1}-\smallgain\estcorrection_k\bigr)
\\
&=
\statemat_k\bigl(\estate_k-\smallgain\correction_k+\smallgain(\correction_k-\estcorrection_k)\bigr)
\\
&=
\statemat_k\bigl(\estate_k-\smallgain\correction_k\bigr)
+
\smallgain\statemat_k\ecorr_k
\\
&=
\statemat_k
\left(
\estate_k
-
\smallgain\infomat_{k|k}^{-1}\knownterm_k
\right)
+
\smallgain\statemat_k\ecorr_k
\\
&=
\statemat_k
\left(
\estate_k
-
\smallgain\infomat_{k|k}^{-1}
\measmat_k^\top
\outcov_k
\measmat_k
\estate_k
\right)
+
\smallgain\statemat_k\ecorr_k
\\
&=
\statemat_k
\left(
\eye{\Nrob d}
-
\smallgain\infomat_{k|k}^{-1}
\measmat_k^\top
\outcov_k
\measmat_k
\right)
\estate_k
+
\smallgain\statemat_k\ecorr_k .
\end{aligned}
\end{equation}
}%

Using \eqref{eq:info_additive_S}, we can rewrite \eqref{eq:err_dyn_full} as
{\small
\begin{equation}\label{eq:new_xtildedyn}
\begin{aligned}
    \estate_{k+1}
&=
\statemat_k\left(
\eye{\Nrob d}
-
\infomat_{k|k}^{-1}
(\infomat_{k|k} - \infomat_{k|k-1})
\right) \estate_k
+
\smallgain\statemat_k
\ecorr_{k}
\\
&=
\underbrace{\statemat_k\infomat_{k|k}^{-1}\infomat_{k|k-1}}_{\define \boldsymbol{\Phi}_k}\estate_k
+
\smallgain\statemat_k
\ecorr_{k}.
\end{aligned}
\end{equation}
}%
%
%
%
%

We now express the ADMM iterations in \eqref{eq:aux_vars_update}--\eqref{eq:closed_form_primal} in compact vector form by stacking the variables of all agents.
The internal state of the ADMM is entirely described by the dual variables $\boldsymbol{q}_{k,\iteradmm} \define \verticalstack{\boldsymbol{q}}{i,k, \iteradmm}{i}{\vertices} \in \nR{2|\edges|d}$, with $\boldsymbol{q}_{i,k, \iteradmm}$ defined in \eqref{eq:def_qk_barbk}.
Likewise, define the stacked primal variable
{\small
\[
\hat{\boldsymbol \xi}_{\neigh,k,\iteradmm}
\define
\verticalstack{\hat{\boldsymbol{\xi}}}{\neigh_i,k,\iteradmm}{i}{\vertices},
\]
}%
together with
{\small
$
\boldsymbol H_{\rho,k}\define\blkdiagconcat{\boldsymbol{H}}{\rho,i,k}{i}{\vertices},
$ 
$
\bar{\knownterm}_k \define \verticalstack{\bar{\knownterm}}{i,k}{i}{\vertices},
$ and
$
\boldsymbol A_q \define \blkdiagconcat{\boldsymbol{A}}{qi}{i}{\vertices},
$
}%
where \(\boldsymbol{H}_{\rho,i,k}\), \(\bar{\knownterm}_{i,k}\), and
\(\boldsymbol{A}_{qi}\) are defined in \eqref{eq:H_decomp},
\eqref{eq:def_qk_barbk}, and \eqref{eq:defAqi}, respectively.
Collecting the dual updates in \eqref{eq:aux_vars_update} into a single stacked
equation for all agents and edges yields
{\small
\[
\boldsymbol q_{k,\iteradmm+1}
=
(1-\alpha)\boldsymbol q_{k,\iteradmm}
+
\alpha\left(
-\boldsymbol P\boldsymbol q_{k,\iteradmm}
+
2\rho\,\boldsymbol A_q \hat{\boldsymbol \xi}_{\neigh,k,\iteradmm}
\right),
\]
}%
where \(\boldsymbol P\) is the permutation matrix exchanging the dual variables
associated with the two endpoints of each edge. On the other hand, by stacking
the closed-form primal update \eqref{eq:closed_form_primal}, one obtains
{\small
\[
\hat{\boldsymbol \xi}_{\neigh,k,\iteradmm}
=
\boldsymbol H_{\rho,k}^{-1}
\left(
\bar{\knownterm}_k+\boldsymbol A_q^\top \boldsymbol q_{k,\iteradmm}
\right).
\]
}%
Substituting this expression into the stacked dual update gives
{\small
\[
\begin{aligned}
\boldsymbol q_{k,\iteradmm+1}
&=
(1-\alpha)\boldsymbol q_{k,\iteradmm}
+
\alpha\Bigl(
-\boldsymbol P\boldsymbol q_{k,\iteradmm}
\\
&\quad 
+
2\rho\,\boldsymbol A_q
\boldsymbol H_{\rho,k}^{-1}
\bigl(
\bar{\knownterm}_k+\boldsymbol A_q^\top \boldsymbol q_{k,\iteradmm}
\bigr)
\Bigr)
\\
&=
(1-\alpha)\boldsymbol q_{k,\iteradmm}
-\alpha \boldsymbol P\boldsymbol q_{k,\iteradmm}
+
2\alpha\rho\,\boldsymbol A_q
\boldsymbol H_{\rho,k}^{-1}
\boldsymbol A_q^\top \boldsymbol q_{k,\iteradmm}
\\
&\quad
+
2\alpha\rho\,\boldsymbol A_q
\boldsymbol H_{\rho,k}^{-1}\bar{\knownterm}_k
\\
&=
\Bigl(
(1-\alpha)\eye{2|\edges|d}
-\alpha \boldsymbol P
+
2\alpha\rho\,\boldsymbol A_q\boldsymbol H_{\rho,k}^{-1}\boldsymbol A_q^\top
\Bigr)\boldsymbol q_{k,\iteradmm}
\\
&\quad
+
2\alpha\rho\,\boldsymbol A_q\boldsymbol H_{\rho,k}^{-1}\bar{\knownterm}_k.
\end{aligned}
\]
}%
Therefore, defining
{\small
\begin{equation}\label{eq:Fdef}
\boldsymbol F_k
\define
\eye{2|\edges|d}
+
\boldsymbol P
-
2\rho\,\boldsymbol A_q\boldsymbol H_{\rho,k}^{-1}\boldsymbol A_q^\top,
\end{equation}
}%
the internal ADMM dynamics can be written in compact form as
{\small
\begin{equation}\label{eq:vect_qdyn}
\boldsymbol q_{k,\iteradmm+1}
=
\underbrace{\left(\eye{2|\edges|d}-\alpha \boldsymbol F_k\right)}_{\define \boldsymbol T_k}
\boldsymbol q_{k,\iteradmm}
+
\alpha\underbrace{2\rho\,\boldsymbol A_q\boldsymbol H_{\rho,k}^{-1}\bar{\knownterm}_k}_{\define \boldsymbol c_k}.
\end{equation}
}%

%
From \eqref{eq:vect_qdyn}, the set of equilibria at time $k$ for the dynamics of ADMM is 
{\small
\[
\mathcal{Q}_{\mathrm{eq},k}
=
\left\{
\boldsymbol{q}_{k}
\;\middle|\;
\boldsymbol{F}_k\boldsymbol{q}_{k}
=
\boldsymbol{c}_k
\right\},
\]
}%
where the linear system in general admits multiple solutions, induced by the edge-based formulation, see \cite[Lemma~2]{bastianello2020asynchronous}.
%
The projection of $\boldsymbol{q}_{k,\iteradmm}$ onto the affine set $\mathcal{Q}_{\mathrm{eq},k}$ is (\cite[Sec.6.2.2]{parikh2014proximal})
{\small
\begin{equation}\label{eq:qeq}
\boldsymbol{q}_{\mathrm{eq},k,\iteradmm} = \boldsymbol{q}_{k,\iteradmm} - \boldsymbol{F}_k^{\dagger}(\boldsymbol{F}_k\boldsymbol{q}_{k,\iteradmm} - \boldsymbol{c}_k)
\end{equation}
}%
where $(\cdot)^\dagger$ denotes the Moore-Penrose pseudo-inverse. Accordingly, define $\tilde{\boldsymbol{q}}_{k,\iteradmm} = \boldsymbol{q}_{\mathrm{eq},k,\iteradmm} - \boldsymbol{q}_{k,\iteradmm}$.
We define the distance from the set of equilibria as 
{\small
\begin{equation}\label{eq:dist_eqset}
    \dist{\mathcal{Q}_{\mathrm{eq},k}}(\boldsymbol{q}_{k,\iteradmm})
\define
\inf_{\boldsymbol{q}_{\ast}\in\mathcal{Q}_{\mathrm{eq},k}}\|\boldsymbol{q}_{k,\iteradmm}-\boldsymbol{q}_{\ast}\|
= \norm{\tilde{\boldsymbol{q}}_{k,\iteradmm}}.
\end{equation}
}%

We now characterize the correction term injected into the observer as an output
of the ADMM dynamics. Specifically, define
{\small
$
\boldsymbol{\Sigma} \define \blkdiagconcat{\boldsymbol{\Sigma}}{i}{i}{\vertices},
$
}%
where \(\boldsymbol{\Sigma}_i\) is a selection matrix such that
{\small
$
\hat{\correction}_{i,k}
=
\boldsymbol{\Sigma}_i
\hat{\boldsymbol{\xi}}_{\neigh_i,k,\maxiteradmm}.
$
}%
Then, by stacking the local corrections, one obtains
{\small
\[
\estcorrection_k
=
\boldsymbol{\Sigma}\hat{\boldsymbol \xi}_{\neigh,k,\maxiteradmm}
=
\boldsymbol{\Sigma}\boldsymbol H_{\rho,k}^{-1}
\left(
\bar{\knownterm}_k+\boldsymbol A_q^\top \boldsymbol q_{k,\maxiteradmm}
\right).
\]
}%

Likewise, if \(\boldsymbol q_{\mathrm{eq},k,\maxiteradmm}\in\mathcal Q_{\mathrm{eq},k}\)
is an equilibrium point of the ADMM dynamics, the corresponding primal variable is
{\small
\[
{\boldsymbol \xi}_{\neigh,k,\mathrm{eq}}
=
\boldsymbol H_{\rho,k}^{-1}
\left(
\bar{\knownterm}_k+\boldsymbol A_q^\top \boldsymbol q_{\mathrm{eq},k,\maxiteradmm}
\right).
\]
}%
By construction of the ADMM reformulation, the equilibrium primal variable
coincides with the exact optimizer of the correction problem. Therefore, the
ideal correction satisfies
{\small
\[
\correction_k
=
\boldsymbol{\Sigma}{\boldsymbol \xi}_{\neigh,k,\mathrm{eq}}
=
\boldsymbol{\Sigma}\boldsymbol H_{\rho,k}^{-1}
\left(
\bar{\knownterm}_k+\boldsymbol A_q^\top \boldsymbol q_{\mathrm{eq},k,\maxiteradmm}
\right).
\]
}%
Recalling the definition of $\ecorr_k$
we obtain
{\small
\[
\begin{aligned}
\ecorr_k
&=
\boldsymbol{\Sigma}\boldsymbol H_{\rho,k}^{-1}
\left(
\bar{\knownterm}_k+\boldsymbol A_q^\top \boldsymbol q_{\mathrm{eq},k,\maxiteradmm}
\right)
-
\boldsymbol{\Sigma}\boldsymbol H_{\rho,k}^{-1}
\left(
\bar{\knownterm}_k+\boldsymbol A_q^\top \boldsymbol q_{k,\maxiteradmm}
\right)
\\
&=
\boldsymbol{\Sigma}\boldsymbol H_{\rho,k}^{-1}\boldsymbol A_q^\top
\tilde{\boldsymbol q}_{k,\maxiteradmm}.
\end{aligned}
\]
}%

The exact coupled dynamics are most conveniently written in terms of the pair
\((\estate_k,\boldsymbol{q}_{k,\iteradmm})\). 
Then the coupled error system is described by
{\small
\begin{subequations}\label{eq:coupled}%
\begin{align}
    \estate_{k+1}
    &=
    \boldsymbol{\Phi}_k\estate_k+\smallgain\statemat_k\ecorr_k,
    \label{eq:err_dyn_compact}
    \\
    \boldsymbol{q}_{k,0}
    &= \boldsymbol{q}_{k-1,\maxiteradmm}
    \\
    \boldsymbol{q}_{k,\iteradmm+1}
    &=
    \boldsymbol{T}_k
    \boldsymbol{q}_{k,\iteradmm}+
    \alpha\boldsymbol{c}_k,
    \label{eq:fast_full}
    \\
    \boldsymbol{q}_{\mathrm{eq},k,\iteradmm}
    &=
    \boldsymbol{q}_{k,\iteradmm}
    -
    \boldsymbol{F}_k^\dagger(\boldsymbol{F}_k\boldsymbol{q}_{k,\iteradmm}-\boldsymbol{c}_k),
    \label{eq:qeq_dyn}
    \\
    \ecorr_k
    &=
    \boldsymbol{\Sigma}\boldsymbol{H}_{\rho,k}^{-1}\boldsymbol{A}_q^\top
    \tilde{\boldsymbol{q}}_{k,\maxiteradmm}.
    \label{eq:output_fast}
\end{align}
\end{subequations}
}

\subsection{ADMM subsystem}

Here, we analyze the stability of the ADMM subsystem while keeping the time
index of the Kalman observer frozen. We denote the frozen observer time by
$\frozeniteradmm$, while the ADMM iterations evolve in the index $\iteradmm$.
Under this assumption, $\bar{\knownterm}_{\frozeniteradmm}$ and
$\infomat_{\frozeniteradmm|\frozeniteradmm}$ are constant, and, from \eqref{eq:vect_qdyn}, the frozen ADMM dynamics reduce to
{\small
\begin{equation}
\label{eq:fast_subsys}
\boldsymbol{q}_{k,\iteradmm+1}
=
\underbrace{
\boldsymbol{T}_{\frozeniteradmm}
\boldsymbol{q}_{k,\iteradmm}+
\alpha\boldsymbol{c}_{\frozeniteradmm}}_{\define \operator{T}_{\frozeniteradmm}(\boldsymbol{q}_{k,\iteradmm})},
\end{equation}
}%
To establish convergence, we recall the following operator-theoretic notions; see, e.g., \cite{ryu2016primer}.
\begin{definition}
    An operator $\operator{T}:\nR{n} \to \nR{n}$ is nonexpansive if  
    {\small $\norm{\operator{T}(\boldsymbol{x})  -\operator{T}(\boldsymbol{y})}
    \leq \norm{\boldsymbol{x} - \boldsymbol{y}}$, $\forall \boldsymbol{x}, \boldsymbol{y} \in \nR{n}
    $}
\end{definition}

\begin{definition}
An operator $\operator{G}:\nR{n}\to\nR{n}$ is \emph{$\alpha$-averaged}, with
$\alpha\in(0,1)$, if there exists a nonexpansive operator
$\operator{T}:\nR{n}\to\nR{n}$ such that
{\small
$
\operator{G}=(1-\alpha)\ideop+\alpha \operator{T}
$,
}
with $\ideop$ being the identity operator.
\end{definition}
For an $\alpha$-averaged operator it holds \cite[Prop.~4.35]{bauschke2011convex}:
{\small
\begin{equation}\label{eq:alpha_avg}
    \norm{\operator{G}\boldsymbol{x} - \operator{G}\boldsymbol{y}}^2 \leq \norm{\boldsymbol{x} - \boldsymbol{y}}^2 - \frac{1-\alpha}{\alpha} \norm{(\ideop-\operator{G})\boldsymbol{x} - (\ideop-\operator{G})\boldsymbol{y}}^2
\end{equation}
}%
\begin{definition}
An operator $\operator{G}:\nR{n}\to\nR{n}$ is said to be \emph{metrically subregular} if there exists $\sigma>0$ such that
{\small
\[
\begin{aligned}
&\fix(\operator{G})\define \{\boldsymbol z\in\nR{n}:\operator{G}(\boldsymbol z)=\boldsymbol z\},
\\
&\dist{\fix(\operator{G})}(\boldsymbol x)
=
\inf_{\boldsymbol z\in\fix(\operator{G})}\|\boldsymbol x-\boldsymbol z\|
\le
\sigma\,\|\boldsymbol x-\operator{G}(\boldsymbol x)\|,
\quad \forall \boldsymbol x\in\nR{n}.
\end{aligned}
\]
}%
\end{definition}
Before proceeding with the stability proof of the ADMM subsystem, we state the following two lemmas, which will be used therein.

The following lemma shows that the nullspace of $\boldsymbol{F}_k$ is purely structural, being determined by the graph-induced permutation matrix $\boldsymbol{P}$ and the edge-based constraint matrix $\boldsymbol{A}_{q}$, and is therefore independent of the time-varying positive definite weights.
\begin{lemma}[Structural invariance of $\ker(\boldsymbol{F}_k)$]
\label{lem:Fk_nullspace}
Let
{\small
$\boldsymbol{F}_k
=
\eye{}+\boldsymbol{P}
-
2\rho\,\boldsymbol{P}\boldsymbol{A}_q\boldsymbol{H}_{\rho,k}^{-1}\boldsymbol{A}_q^\top$
}%
as in \eqref{eq:Fdef}, where $\boldsymbol{P}$ and $\boldsymbol{A}_q$ are constant matrices induced by the graph structure, and ${\boldsymbol{H}_{\rho,k}\succ0}$ is block diagonal for all $k$. Then
{\small
\[
\ker(\boldsymbol{F}_k)
=
\ker(\eye{}+\boldsymbol{P})\cap\ker(\boldsymbol{A}_q^\top),
\qquad \forall k.
\]
}%
In particular, $\ker(\boldsymbol{F}_k)$ is independent of $k$.
\end{lemma}

\begin{proof}
Consider $\boldsymbol{T}_k$ defined as in \eqref{eq:vect_qdyn}. One has
{\small
$\ker(\boldsymbol{F}_k)=\ker(\eye{}-\boldsymbol{T}_k)$.
}%
By \cite[Lemma~2]{bastianello2020asynchronous},
{\small
\[
\ker(\eye{}-\boldsymbol{T}_k)\subseteq \ker(\boldsymbol{A}_q^\top).
\]
}%
Hence, if $\boldsymbol{v}\in\ker(\boldsymbol{F}_k)$, then
$\boldsymbol{A}_q^\top\boldsymbol{v}=0$. Substituting into the definition of
$\boldsymbol{F}_k$ yields
{\small
\[
\boldsymbol{0}
=
\boldsymbol{F}_k\boldsymbol{v}
=
(\eye{}+\boldsymbol{P})\boldsymbol{v}
-
2\rho\,\boldsymbol{P}\boldsymbol{A}_q\boldsymbol{H}_{\rho,k}^{-1}\boldsymbol{A}_q^\top\boldsymbol{v}
=
(\eye{}+\boldsymbol{P})\boldsymbol{v},
\]
}%
so that $\boldsymbol{v}\in\ker(\eye{}+\boldsymbol{P})$. Therefore,
{\small
\[
\ker(\boldsymbol{F}_k)
\subseteq
\ker(\eye{}+\boldsymbol{P})\cap\ker(\boldsymbol{A}_q^\top).
\]
}%

Conversely, let
$\boldsymbol{v}\in\ker(\eye{}+\boldsymbol{P})\cap\ker(\boldsymbol{A}_q^\top)$.
Then
{\small
\[
(\eye{}+\boldsymbol{P})\boldsymbol{v}=0,
\qquad
\boldsymbol{A}_q^\top\boldsymbol{v}=0,
\]
}%
and thus
{\small
\[
\boldsymbol{F}_k\boldsymbol{v}
=
(\eye{}+\boldsymbol{P})\boldsymbol{v}
-
2\rho\,\boldsymbol{P}\boldsymbol{A}_q\boldsymbol{H}_{\rho,k}^{-1}\boldsymbol{A}_q^\top\boldsymbol{v}
=
\boldsymbol{0}.
\]
}%
Hence,
{\small
\[
\ker(\eye{}+\boldsymbol{P})\cap\ker(\boldsymbol{A}_q^\top)
\subseteq
\ker(\boldsymbol{F}_k).
\]
}%
Combining the two inclusions proves the claim. Since neither
$\boldsymbol{P}$ nor $\boldsymbol{A}_q$ depends on $k$, the right-hand side is independent of $k$, and therefore so is $\ker(\boldsymbol{F}_k)$.
\end{proof}
\begin{lemma}[Uniform boundedness of the pseudoinverse of $\boldsymbol F_k$]
\label{lem:Fdagger_bound}
Suppose that the penalty parameter $\rho>0$ is fixed. Under
Assumptions~\ref{ass:boundedness}--\ref{ass:invertibility}, there exists a constant $\bar F^\dagger>0$ such that
\begin{equation}
\label{eq:Fdagger_uniform_bound}
\|\boldsymbol F_k^\dagger\|
\le
\bar F^\dagger,
\qquad \forall k.
\end{equation}
Equivalently, the smallest nonzero singular value of $\boldsymbol F_k$
admits a uniform positive lower bound.
\end{lemma}

\begin{proof}
We proceed in four steps.

\emph{Step 1: uniform positive definiteness of $\boldsymbol H_{\rho,i,k}$.}
By construction,
{\small
$
\boldsymbol H_{\rho,i,k}
=
\rho \boldsymbol H_{0,i}
+
\varepsilon \bar{\boldsymbol H}_{1,i,k}.
$
}%
Since $\boldsymbol H_{0,i}\succ 0$ and the information contributions entering
$\bar{\boldsymbol H}_{1,i,k}$ are positive semidefinite, one has
{\small
$
\boldsymbol H_{\rho,i,k}\succeq \rho \boldsymbol H_{0,i}.
$
}%
Hence, letting
{\small
\[
\underline h_0 := \min_{i\in\mathcal V}\lambda_{\min}(\boldsymbol H_{0,i}) >0,
\]
}%
it follows that
{\small
\[
\boldsymbol H_{\rho,i,k}\succeq \rho \underline h_0 \, \eye{},
\qquad \forall i,\ \forall k,
\]
}%
and therefore
{\small
\begin{equation}
\label{eq:Hrho_inv_bound_local}
\|\boldsymbol H_{\rho,i,k}^{-1}\|
\le
\frac{1}{\rho \underline h_0},
\qquad \forall i,\ \forall k.
\end{equation}
}%
Since $\boldsymbol H_{\rho,k}$ is block diagonal, \eqref{eq:Hrho_inv_bound_local} implies
{\small
\begin{equation}
\label{eq:Hrho_inv_bound_global}
\|\boldsymbol H_{\rho,k}^{-1}\|
\le
\frac{1}{\rho \underline h_0},
\qquad \forall k.
\end{equation}
}%

\emph{Step 2: boundedness of the family $\{\boldsymbol F_k\}_k$.}
Using the definition of $\boldsymbol F_k$ and \eqref{eq:Hrho_inv_bound_global},
{\small
\[
\begin{aligned}
\|\boldsymbol F_k\|
&\le
\|\eye{}+\boldsymbol P\|
+
2\rho \|\boldsymbol P\|\,\|\boldsymbol A_q\|^2
\|\boldsymbol H_{\rho,k}^{-1}\|
\\
&\le
\|\eye{}+\boldsymbol P\|
+
\frac{2\|\boldsymbol P\|\,\|\boldsymbol A_q\|^2}{\underline h_0},
\end{aligned}
\]
}%
for all $k$. Therefore, the set $\mathcal F\define\{\boldsymbol F_k\}_k$ is bounded in
$\mathbb{R}^{2|\edges|d \times 2|\edges|d}$. Since this space is finite-dimensional,
its closure $\overline{\mathcal F}$ is compact.

\emph{Step 3: constant rank of $\boldsymbol F_k$.}
By Lemma~\ref{lem:Fk_nullspace}, one has
{\small
\[
\ker(\boldsymbol F_k)
=
\ker(\eye{}+\boldsymbol P)\cap\ker(\boldsymbol A_q^\top),
\qquad \forall k.
\]
}%
Hence $\ker(\boldsymbol F_k)$ is independent of $k$, and so
$\operatorname{rank}(\boldsymbol F_k)$ is constant over $\mathcal F$ and over
its closure $\overline{\mathcal F}$.

\emph{Step 4: uniform lower bound on the smallest nonzero singular value.}
Let $\sigma_{\min}^+(\boldsymbol F)$ denote the smallest nonzero singular value
of a matrix $\boldsymbol F$. On the set of matrices of fixed rank,
$\sigma_{\min}^+(\cdot)$ is continuous. Since every matrix in
$\overline{\mathcal F}$ has the same rank as $\boldsymbol F_k$, the map
{\small
\[
\boldsymbol F \mapsto \sigma_{\min}^+(\boldsymbol F)
\]
}%
is continuous on the compact set $\overline{\mathcal F}$. Moreover,
$\sigma_{\min}^+(\boldsymbol F)>0$ for every $\boldsymbol F\in\overline{\mathcal F}$
because the rank is constant. Therefore, by compactness, there exists
$\underline \sigma_F>0$ such that
{\small
\[
\sigma_{\min}^+(\boldsymbol F_k)\ge \underline \sigma_F,
\qquad \forall k.
\]
}%
Finally, using the standard identity
{\small
\[
\|\boldsymbol F_k^\dagger\|
=
\frac{1}{\sigma_{\min}^+(\boldsymbol F_k)},
\]
}%
we obtain
{\small
\[
\|\boldsymbol F_k^\dagger\|
\le
\frac{1}{\underline \sigma_F}
=: \bar F^\dagger,
\qquad \forall k,
\]
}%
which proves \eqref{eq:Fdagger_uniform_bound}.
\end{proof}


\begin{lemma}[Uniform exponential stability of the ADMM subsystem]
\label{lem:bl}
Under Assumptions~\ref{ass:boundedness}--\ref{ass:invertibility}, for every fixed
$\bar{\knownterm}_{\frozeniteradmm}$ and $\infomat_{\frozeniteradmm|\frozeniteradmm}$, the frozen ADMM
iteration \eqref{eq:fast_subsys} satisfies
{\small
\begin{equation}
\label{eq:qexpstab1step}
\dist{\mathcal{Q}_{\mathrm{eq},{\frozeniteradmm}}}(\boldsymbol{q}_{k,\iteradmm+1})
\le
\mu\,\dist{\mathcal{Q}_{\mathrm{eq},{\frozeniteradmm}}}(\boldsymbol{q}_{k,\iteradmm}),
\qquad
\mu
=
\sqrt{1-\frac{1-\alpha}{\alpha\sigma^2}}
\in(0,1).
\end{equation}
}%
In particular, after \(\maxiteradmm\geq 1\) ADMM iterations performed during observer step \(\frozeniteradmm\),
{\small
\begin{equation}\label{eq:qexpstab}
    \dist{\mathcal{Q}_{\mathrm{eq},{\frozeniteradmm}}}(\boldsymbol{q}_{k,\maxiteradmm})
\le
\mu^\maxiteradmm\,\dist{\mathcal{Q}_{\mathrm{eq},{\frozeniteradmm}}}(\boldsymbol{q}_{k,0})
\end{equation}
}%
\end{lemma}

\begin{proof}
The following argument is a specialization of \cite[Thm.~3]{bastianello2024robust}
to the present deterministic affine setting.

The operator $\operator{T}_{\frozeniteradmm}$ in \eqref{eq:fast_subsys} corresponds to the
relaxed Peaceman--Rachford splitting applied to the dual of
\eqref{eq:admm_global_problem} \cite{bastianello2020asynchronous}. Since the local
costs $J_i(\extloccorr{i})$ are convex, $\operatorname{T}_{\frozeniteradmm}$ is
$\alpha$-averaged \cite[Sec.~7.3]{ryu2016primer}. Moreover, by construction,
{\small$
\fix(\operatorname{T}_{\frozeniteradmm})=\mathcal{Q}_{\mathrm{eq},\frozeniteradmm}.
$}

We next show that $\operatorname{T}_{\frozeniteradmm}$ is metrically subregular.
Using \eqref{eq:qeq} and \eqref{eq:dist_eqset}, we obtain
{\small
\[
\dist{\mathcal Q_{\mathrm{eq},\frozeniteradmm}}(\boldsymbol{q}_{k,\iteradmm})
=
\|\boldsymbol{q}_{k,\iteradmm}-\boldsymbol{q}_{\mathrm{eq},\frozeniteradmm,\iteradmm}\|
=
\|\boldsymbol F_{\frozeniteradmm}^\dagger
(\boldsymbol F_{\frozeniteradmm}\boldsymbol{q}_{k,\iteradmm}-\boldsymbol c_{\frozeniteradmm})\|.
\]}%
Therefore, using the submultiplicativity of the norm gives
{\small
\[
\dist{\mathcal Q_{\mathrm{eq},\frozeniteradmm}}(\boldsymbol{q}_{k,\iteradmm})
\le
\|\boldsymbol F_{\frozeniteradmm}^\dagger\|
\,
\|\boldsymbol F_{\frozeniteradmm}\boldsymbol{q}_{k,\iteradmm}-\boldsymbol c_{\frozeniteradmm}\|.
\]
}%
Now, from the definition of $\operatorname{T}_{\frozeniteradmm}$ in \eqref{eq:fast_subsys} and the one of $\boldsymbol{T}_{\frozeniteradmm}$ in \eqref{eq:vect_qdyn}, {\small$\boldsymbol F_{\frozeniteradmm}\boldsymbol{q}_{k,\iteradmm} - \boldsymbol c_{\frozeniteradmm} = -\frac{1}{\alpha} (\operatorname{T}_{\frozeniteradmm}(\boldsymbol{q}_{k,\iteradmm})-\boldsymbol{q}_{k,\iteradmm})$}, and hence
{\small
\[
\dist{\mathcal Q_{\mathrm{eq},\frozeniteradmm}}(\boldsymbol{q}_{k,\iteradmm})
\le
\frac{\|\boldsymbol F_{\frozeniteradmm}^\dagger\|}{\alpha}
\|\boldsymbol{q}_{k,\iteradmm}-\operator{T}_{\frozeniteradmm}(\boldsymbol{q}_{k,\iteradmm})\|.
\]
}%
This shows that $\operatorname{T}_{\frozeniteradmm}$ is metrically subregular for some
{\small
$\sigma_{\frozeniteradmm} \ge \frac{\|\boldsymbol F_{\frozeniteradmm}^\dagger\|}{\alpha}$.
}
By Lemma~\ref{lem:Fdagger_bound}, the pseudoinverse of $\boldsymbol F_{\frozeniteradmm}$
is uniformly bounded. Hence a common subregularity constant can be chosen as
{\small
\[
\sigma
\ge
\frac{\bar F^\dagger}{\alpha},
\qquad
\bar F^\dagger
=
\sup_{\frozeniteradmm}\|\boldsymbol F_{\frozeniteradmm}^\dagger\|
<
\infty.
\]
}%
Therefore, the metric subregularity estimate holds uniformly with respect to the
frozen observer index~$\frozeniteradmm$.

Since $\operatorname{T}_{\frozeniteradmm}$ is $\alpha$-averaged, by \eqref{eq:alpha_avg},
for every $\boldsymbol{q}_{\ast,\frozeniteradmm}\in\mathcal Q_{\mathrm{eq},\frozeniteradmm}$, 
one has
{\small
\[
\|\operatorname{T}_{\frozeniteradmm}(\boldsymbol{q}_{k,\iteradmm})-\boldsymbol{q}_{\ast,\frozeniteradmm}\|^2
\le
\|\boldsymbol{q}_{k,\iteradmm}-\boldsymbol{q}_{\ast,\frozeniteradmm}\|^2
-
\frac{1-\alpha}{\alpha}
\|\operatorname{T}_{\frozeniteradmm}(\boldsymbol{q}_{k,\iteradmm})-\boldsymbol{q}_{k,\iteradmm}\|^2.
\]}%
Choosing $\boldsymbol q_{\ast,\frozeniteradmm}=\boldsymbol q_{\mathrm{eq},\frozeniteradmm}$ and
using the definition of the distance to the equilibrium set gives
{\small
\[
\dist{\mathcal Q_{\mathrm{eq},\frozeniteradmm}}(\boldsymbol q_{k,\iteradmm+1})^2
\le
\dist{\mathcal Q_{\mathrm{eq},\frozeniteradmm}}(\boldsymbol{q}_{k,\iteradmm})^2
-
\frac{1-\alpha}{\alpha}
\|\operatorname{T}_{\frozeniteradmm}(\boldsymbol{q}_{k,\iteradmm})-\boldsymbol{q}_{k,\iteradmm}\|^2.
\]
}%
Using metric subregularity,
{\small
$\|\operatorname{T}_{\frozeniteradmm}(\boldsymbol{q}_{k,\iteradmm})-\boldsymbol{q}_{k,\iteradmm}\|
\ge
\frac{1}{\sigma}\,
\dist{\mathcal Q_{\mathrm{eq},\frozeniteradmm}}(\boldsymbol{q}_{k,\iteradmm})$}.
Substituting into the previous inequality yields
{\small
\[
\dist{\mathcal Q_{\mathrm{eq},\frozeniteradmm}}(\boldsymbol{q}_{k,\iteradmm+1})^2
\le
\left(
1-\frac{1-\alpha}{\alpha\sigma^2}
\right)
\dist{\mathcal Q_{\mathrm{eq},\frozeniteradmm}}(\boldsymbol{q}_{k,\iteradmm})^2.
\]}%
where notice that, since $\alpha\in(0,1)$ and $\sigma$ can be chosen large enough, the term in parentheses takes values in $(0,1)$ for $\sigma$ sufficiently large. 
Taking square roots gives \eqref{eq:qexpstab1step}.

Finally, \eqref{eq:qexpstab} follows by iterating \eqref{eq:qexpstab1step} \(\maxiteradmm\) times.

\end{proof}

\subsection{Kalman-like subsystem}

We now consider the reduced system obtained for exact correction,
i.e., $\ecorr_k\equiv 0$ and expanding $\boldsymbol{\Phi}_k$:
{\small
\begin{equation}
\label{eq:reduced_system}
\estate_{k+1}=\statemat_k\infomat_{k|k}^{-1}
\infomat_{k|k-1}\estate_k.
\end{equation}
}%

\begin{lemma}[UGES of the Kalman-like subsystem]
\label{lem:reduced_uges}
Under Assumptions~\ref{ass:boundedness}, \ref{ass:compobs}, and
\ref{ass:invertibility}, the reduced system \eqref{eq:reduced_system} is
uniformly globally exponentially stable. In particular, the Lyapunov function
\[
\underline{s}_\smallgain\|\estate_k\|^2 \le V_k \define \estate_k^\top \infomat_{k|k-1}\estate_k \le \bar{s}_\smallgain\|\estate_k\|^2
\]
satisfies
\[
V_{k+1} - V_k \leq -(1-\ffactor )\underline{s}_\smallgain\norm{\estate_k}^2 .
\]
\end{lemma}

\begin{proof}
Using \eqref{eq:reduced_system},
{\small
\begin{equation}\label{eq:Vk+1}
\begin{aligned}
V_{k+1}
&=
\estate_{k+1}^\top \infomat_{k+1|k}\estate_{k+1}
\\
&=
\left(
\statemat_k\infomat_{k|k}^{-1}\infomat_{k|k-1}\estate_k
\right)^\top
\infomat_{k+1|k}
\left(
\statemat_k\infomat_{k|k}^{-1}\infomat_{k|k-1}\estate_k
\right)
\\
&=
\estate_k^\top
\infomat_{k|k-1}\infomat_{k|k}^{-1}
\statemat_k^\top
\infomat_{k+1|k}
\statemat_k
\infomat_{k|k}^{-1}
\infomat_{k|k-1}
\estate_k.
\end{aligned}
\end{equation}
}%
Using the modified prediction \eqref{eq:forgetting_prediction_clean}, {\small$\statemat_k^\top \infomat_{k+1|k}\statemat_k = \ffactor \infomat_{k|k}$.}
Then substituting back in \eqref{eq:Vk+1}, we obtain 
{\small
\begin{equation}
    V_{k+1} \leq \ffactor  \estate_k^\top \infomat_{k|k-1}\infomat_{k|k}^{-1}
\infomat_{k|k-1} \estate_k
\end{equation}
}%
Since the correction step adds positive semidefinite information, one has {\small$\infomat_{k|k}\succeq \infomat_{k|k-1}$}. Because matrix inversion reverses the Loewner order for positive definite matrices,
{\small
$
\infomat_{k|k}^{-1}
\preceq
\infomat_{k|k-1}^{-1}.
$
}
Premultiplying and postmultiplying by $\infomat_{k|k-1}$ gives
{\small
\[
\infomat_{k|k-1}\infomat_{k|k}^{-1}\infomat_{k|k-1}
\preceq
\infomat_{k|k-1}\infomat_{k|k-1}^{-1}\infomat_{k|k-1}
=
\infomat_{k|k-1}.
\]
}%
Hence,
{\small
\[
V_{k+1}
\le
\ffactor\,
\estate_k^\top \infomat_{k|k-1}\estate_k
=
\ffactor V_k.
\]
}%
It follows that
{\small
\[
V_{k+1}-V_k
\le
-(1-\ffactor)V_k.
\]
}%
Finally, by Lemma~\ref{lemma:conditioning},
{\small
\[
V_k\ge \underline{s}_\smallgain \|\estate_k\|^2,
\]
}%
and therefore
{\small
\[
V_{k+1}-V_k
\le
-(1-\ffactor)\underline{s}_\smallgain \|\estate_k\|^2,
\]
}%
which proves the claim.

\end{proof}

In Appendix~\ref{app:matrixffuges}, we extend the above UGES result to the case of a matrix forgetting factor. The associated sufficient conditions are, however, more difficult to verify a priori and generally conservative.

\begin{remark}[Practical tuning of Matrix Forgetting Factor]
The sufficient conditions in Lemma~\ref{lem:reduced_uges_matrix} (in Appendix~\ref{app:matrixffuges}) are primarily of theoretical interest, as they depend on the conditioning of the time-varying information matrix and may be quite conservative in practice. Accordingly, a matrix forgetting factor is more naturally chosen by tuning, e.g., by assigning faster decay rates to state components expected to be more strongly affected by noise and then validating the resulting behavior numerically.
\end{remark}

\subsection{Stability of the feedback interconnection}
\label{subsec:interconnection}

We now analyze the feedback interconnection between the estimation-error dynamics \eqref{eq:new_xtildedyn} and the ADMM correction dynamics \eqref{eq:vect_qdyn} when
\(\maxiteradmm\geq 1\) ADMM iterations are performed at each observer step, as described in \eqref{eq:coupled}. The proof proceeds in four steps.
First, we derive an ISS-type estimate for the reduced dynamics with respect to the correction mismatch.
Second, we bound the drift of the time-varying equilibrium set of the ADMM subsystem.
Third, we combine this drift estimate with the \(\maxiteradmm\)-step contraction induced by Lemma~\ref{lem:bl} to obtain a perturbed estimate for the ADMM subsystem.
Fourth, we combine the perturbed estimates of the reduced and ADMM dynamics into a comparison system for the full interconnection, from which exponential stability follows by a small-gain argument.

\begin{theorem}[UGES of the interconnected system]
\label{thm:main}
Under Assumptions~\ref{ass:boundedness}, \ref{ass:compobs}, and
\ref{ass:invertibility}, suppose that
{\small
\begin{equation}
\label{eq:smallgain_bound}
\smallgain
<
\frac{(1-\sqrt{\ffactor})(1-\mu^{\maxiteradmm})}
{c_{12}c_{21}+(1-\sqrt{\ffactor})c_{22}} .
\end{equation}
}
Then the set
{\small
\[
\{(\estate_k,\boldsymbol{q}_{k,0})\mid
\estate_k=\zeros{},\ 
\boldsymbol{q}_{k,0}\in\mathcal Q_{\mathrm{eq},k}\}
\]
}%
of the coupled system is uniformly globally exponentially stable.
In particular, there exist constants \(C>0\) and \(\lambda\in(0,1)\) such that
{\small
\[
\|\estate_k\|
+
\dist{\mathcal Q_{\mathrm{eq},k}}(\boldsymbol{q}_{k,0})
\le
C\lambda^k
\bigl(
\|\estate_0\|
+
\dist{\mathcal Q_{\mathrm{eq},0}}(\boldsymbol{q}_{0,0})
\bigr).
\]
}%
\end{theorem}

We next establish the auxiliary estimates needed in the proof of
Theorem~\ref{thm:main}.

\begin{lemma}[Perturbed estimate for the reduced dynamics]
\label{lem:slow_iss}
Define
{\small
\[
w_k \define \sqrt{V_k}
=
\sqrt{\estate_k^\top \infomat_{k|k-1}\estate_k}.
\]
}%
Then there exists a constant \(c_w>0\) such that
{\small
\begin{equation}
\label{eq:w_iss_final_rewritten}
w_{k+1}
\le
\sqrt{\ffactor}\,w_k
+
\smallgain c_w\,
\dist{\mathcal Q_{\mathrm{eq},k}}(\boldsymbol{q}_{k,\maxiteradmm}),
\qquad \forall k,
\end{equation}
}%
where
{\small
\[
c_w \define \sqrt{\bar{s}_\smallgain}\,\bar a\,c_h.
\]
}%
\end{lemma}

\begin{proof}
By Lemma~\ref{lem:reduced_uges}, the unperturbed reduced system satisfies
{\small
\[
w_{k+1}\le \sqrt{\ffactor}\,w_k.
\]
}%
For the full interconnected system, \eqref{eq:err_dyn_compact} and the
triangle inequality yield
{\small
\[
w_{k+1}
\le
\|\statemat_k\infomat_{k|k}^{-1}\infomat_{k|k-1}\estate_k\|_{\infomat_{k+1|k}}
+
\smallgain\|\statemat_k\ecorr_k\|_{\infomat_{k+1|k}}.
\]
}%
The first term is bounded by \(\sqrt{\ffactor}\,w_k\).
For the second term, since
{\small
\[
\infomat_{k+1|k}\preceq \bar s_\smallgain \eye{},
\qquad
\|\statemat_k\|\le \bar a,
\]
}%
one has
{\small
\[
\|\statemat_k\ecorr_k\|_{\infomat_{k+1|k}}
\le
\sqrt{\bar s_\smallgain}\,\bar a\,\|\ecorr_k\|.
\]
}%
Moreover, by \eqref{eq:output_fast} and the uniform boundedness of
\(\boldsymbol H_{\rho,k}^{-1}\), there exists \(c_h>0\) such that
{\small
\[
\|\ecorr_k\|
\le
c_h\|\tilde{\boldsymbol q}_{k,\maxiteradmm}\|.
\]
}%
Using \eqref{eq:dist_eqset} gives
{\small
\[
w_{k+1}
\le
\sqrt{\ffactor}\,w_k
+
\smallgain \sqrt{\bar s_\smallgain}\,\bar a\,c_h\,
\dist{\mathcal Q_{\mathrm{eq},k}}(\boldsymbol{q}_{k,\maxiteradmm}),
\]
}%
which is \eqref{eq:w_iss_final_rewritten}.
\end{proof}

We now turn to the ADMM subsystem.
Because the distributed optimization layer solves a time-varying problem,
the natural quantity to track is the distance to the equilibrium set
\(\mathcal Q_{\mathrm{eq},k}\), rather than the error with respect to a
single equilibrium point. Since the observer interacts with the ADMM layer
only once per sampling step, we measure this distance at the beginning of
each observer step, namely at \(\boldsymbol q_{k,0}\).

Let \(\boldsymbol{w}_k\in\mathcal{Q}_{\mathrm{eq},k}\) satisfy
{\small
\[
\dist{\mathcal{Q}_{\mathrm{eq},k}}(\boldsymbol{q}_{k+1,0})
=
\|\boldsymbol{q}_{k+1,0}-\boldsymbol{w}_k\|.
\]
}%
Then, by the triangle inequality,
{\small
\[
\dist{\mathcal{Q}_{\mathrm{eq},k+1}}(\boldsymbol{q}_{k+1,0})
\le
\|\boldsymbol{q}_{k+1,0}-\boldsymbol{w}_k\|
+
\dist{\mathcal{Q}_{\mathrm{eq},k+1}}(\boldsymbol{w}_k).
\]
}%
Hence
{\small
\begin{equation}
\label{eq:dk_triangle}
\dist{\mathcal{Q}_{\mathrm{eq},k+1}}(\boldsymbol{q}_{k+1,0})
\le
\dist{\mathcal{Q}_{\mathrm{eq},k}}(\boldsymbol{q}_{k+1,0})
+
\Delta_{\mathcal{Q},k},
\end{equation}
}%
where
{\small
\[
\Delta_{\mathcal{Q},k}
\define
\sup_{\boldsymbol{w}\in\mathcal{Q}_{\mathrm{eq},k}}
\dist{\mathcal{Q}_{\mathrm{eq},k+1}}(\boldsymbol{w})
\]
}%
denotes the one-sided drift of the equilibrium set.

The next result bounds this drift in terms of the reduced-state error and
the current distance from equilibrium.

\begin{lemma}[Drift of the equilibrium set]
\label{lem:set_drift}
There exist constants \(c_{Q,w}>0\) and \(c_{Q,q}>0\) such that
{\small
\begin{equation}
\label{eq:DeltaQ_bound_final_rewritten}
\Delta_{\mathcal Q,k}
\le
c_{Q,w}\,w_k
+
\smallgain c_{Q,q}\,
\dist{\mathcal Q_{\mathrm{eq},k}}(\boldsymbol{q}_{k,0}),
\qquad \forall k.
\end{equation}
}%
In particular, one may choose
{\small
\[
\begin{aligned}
c_{Q,w}
&\define
\bar{F}^\dagger
\left(
\frac{\bar{F}_\Delta \bar{F}^\dagger \bar{M}}{\sqrt{\underline{s}_\smallgain}}
+
\frac{\bar{M}(1+c_\Phi)+\bar{M}_\Delta}{\sqrt{\underline{s}_\smallgain}}
\right),
\\[1mm]
c_{Q,q}
&\define
\bar{F}^\dagger \bar{M}\bar a c_h.
\end{aligned}
\]
}%
where
{\small
\[
\bar F^\dagger \define \sup_k \|\boldsymbol F_k^\dagger\| < \infty,
\qquad
\bar F_\Delta \define \sup_k \|\boldsymbol F_{k+1}-\boldsymbol F_k\| < \infty.
\]
}%
\end{lemma}

\begin{proof}
Define the minimum-norm equilibrium point
{\small
\[
\boldsymbol{q}_{\mathrm{ref},k}
\define
\boldsymbol{F}_k^\dagger \boldsymbol{c}_k.
\]
}%
Since \(\ker(\boldsymbol{F}_k)\) is independent of \(k\) by
Lemma~\ref{lem:Fk_nullspace}, the equilibrium sets satisfy
{\small
\[
\mathcal{Q}_{\mathrm{eq},k}
=
\{\boldsymbol{q}_{\mathrm{ref},k}+\boldsymbol{v}\mid \boldsymbol{v}\in\ker(\boldsymbol{F}_k)\},
\]
}%
that is, they are affine translates of the same linear subspace.
Therefore, for every \(\boldsymbol{w}\in\mathcal{Q}_{\mathrm{eq},k}\),
there exists \(\boldsymbol{v}\in\ker(\boldsymbol{F}_k)=\ker(\boldsymbol{F}_{k+1})\)
such that
{\small
\[
\boldsymbol{w}=\boldsymbol{q}_{\mathrm{ref},k}+\boldsymbol{v}.
\]
}%
Since \(\mathcal{Q}_{\mathrm{eq},k+1}\) is invariant under translations by
\(\ker(\boldsymbol{F}_{k+1})\), it follows that
{\small
\[
\dist{\mathcal{Q}_{\mathrm{eq},k+1}}(\boldsymbol{w})
=
\dist{\mathcal{Q}_{\mathrm{eq},k+1}}(\boldsymbol{q}_{\mathrm{ref},k}).
\]
}%
Hence
{\small
\[
\Delta_{\mathcal{Q},k}
=
\dist{\mathcal{Q}_{\mathrm{eq},k+1}}(\boldsymbol{q}_{\mathrm{ref},k}).
\]
}%

Now, by \eqref{eq:dist_eqset} applied to the set
{\small
\[
\mathcal{Q}_{\mathrm{eq},k+1}
=
\{\boldsymbol{q}\mid \boldsymbol{F}_{k+1}\boldsymbol{q}=\boldsymbol{c}_{k+1}\},
\]
}%
we have
{\small
\[
\Delta_{\mathcal{Q},k}
=
\left\|
\boldsymbol{F}_{k+1}^\dagger
\bigl(
\boldsymbol{F}_{k+1}\boldsymbol{q}_{\mathrm{ref},k}-\boldsymbol{c}_{k+1}
\bigr)
\right\|.
\]
}%
Using \(\boldsymbol{F}_k\boldsymbol{q}_{\mathrm{ref},k}=\boldsymbol{c}_k\),
we obtain
{\small
\[
\boldsymbol{F}_{k+1}\boldsymbol{q}_{\mathrm{ref},k}-\boldsymbol{c}_{k+1}
=
(\boldsymbol{F}_{k+1}-\boldsymbol{F}_k)\boldsymbol{q}_{\mathrm{ref},k}
-
(\boldsymbol{c}_{k+1}-\boldsymbol{c}_k),
\]
}%
and therefore
{\small
\begin{equation}
\label{eq:DeltaQ_pre}
\Delta_{\mathcal{Q},k}
\le
\|\boldsymbol{F}_{k+1}^\dagger\|
\Bigl(
\|\boldsymbol{F}_{k+1}-\boldsymbol{F}_k\|\,\|\boldsymbol{q}_{\mathrm{ref},k}\|
+
\|\boldsymbol{c}_{k+1}-\boldsymbol{c}_k\|
\Bigr).
\end{equation}
}%

Since \(\bar{\knownterm}_k\) depends linearly on \(\estate_k\), there exists a
matrix \(\boldsymbol{M}_k\) such that
{\small
\[
\boldsymbol{c}_k=\boldsymbol{M}_k\estate_k,
\]
}%
with \(\|\boldsymbol{M}_k\|\le \bar{M}\) uniformly in \(k\).
Hence
{\small
\begin{equation}
\label{eq:ck_bound}
\|\boldsymbol{c}_k\|
\le
\bar{M}\,\|\estate_k\|
\le
\frac{\bar{M}}{\sqrt{\underline{s}_\smallgain}}\,w_k.
\end{equation}
}%
It follows that
{\small
\begin{equation}
\label{eq:qref_bound}
\|\boldsymbol{q}_{\mathrm{ref},k}\|
\le
\|\boldsymbol{F}_k^\dagger\|\,\|\boldsymbol{c}_k\|
\le
\frac{\bar{F}^\dagger \bar{M}}{\sqrt{\underline{s}_\smallgain}}\,w_k.
\end{equation}
}%

Moreover,
{\small
\[
\boldsymbol{c}_{k+1}-\boldsymbol{c}_k
=
\boldsymbol{M}_{k+1}(\estate_{k+1}-\estate_k)
+
(\boldsymbol{M}_{k+1}-\boldsymbol{M}_k)\estate_k.
\]
}%
Since the matrices \(\boldsymbol{M}_k\) are uniformly bounded and have
uniformly bounded variation, there exists \(\bar{M}_\Delta>0\) such that
{\small
\begin{equation}
\label{eq:ckdiff_pre}
\|\boldsymbol{c}_{k+1}-\boldsymbol{c}_k\|
\le
\bar{M}\,\|\estate_{k+1}-\estate_k\|
+
\bar{M}_\Delta\,\|\estate_k\|.
\end{equation}
}%

From \eqref{eq:err_dyn_compact},
{\small
\[
\estate_{k+1}-\estate_k
=
(\boldsymbol{\Phi}_k-\eye{\Nrob d})\estate_k
+
\smallgain\statemat_k\ecorr_k,
\]
}%
hence
{\small
\[
\|\estate_{k+1}-\estate_k\|
\le
(1+\|\boldsymbol{\Phi}_k\|)\,\|\estate_k\|
+
\smallgain\|\statemat_k\|\,\|\ecorr_k\|.
\]
}%
Using the uniform bounds on \(\boldsymbol{\Phi}_k\) and \(\statemat_k\),
together with
{\small
\[
\|\ecorr_k\|
\le
c_h\,\dist{\mathcal Q_{\mathrm{eq},k}}(\boldsymbol{q}_{k,\maxiteradmm}),
\]
}%
there exists \(c_\Phi>0\) such that
{\small
\begin{equation}
\label{eq:xdiff_bound}
\|\estate_{k+1}-\estate_k\|
\le
\frac{1+c_\Phi}{\sqrt{\underline{s}_\smallgain}}\,w_k
+
\smallgain\,\bar{a}\,c_h\,\dist{\mathcal Q_{\mathrm{eq},k}}(\boldsymbol{q}_{k,\maxiteradmm}).
\end{equation}
}%

By \eqref{eq:qexpstab},
{\small
\[
\dist{\mathcal Q_{\mathrm{eq},k}}(\boldsymbol{q}_{k,\maxiteradmm})
\le
\mu^{\maxiteradmm}
\dist{\mathcal Q_{\mathrm{eq},k}}(\boldsymbol{q}_{k,0}),
\]
}%
hence this term is bounded by
{\small
$
\dist{\mathcal Q_{\mathrm{eq},k}}(\boldsymbol{q}_{k,0}),
$
}%
up to a multiplicative factor strictly smaller than one. Absorbing this factor into the constant \(c_{Q,q}\), we obtain from \eqref{eq:xdiff_bound}
{\small
\begin{equation}
\label{eq:ckdiff_bound}
\|\boldsymbol{c}_{k+1}-\boldsymbol{c}_k\|
\le
\left(
\frac{\bar{M}(1+c_\Phi)+\bar{M}_\Delta}{\sqrt{\underline{s}_\smallgain}}
\right) w_k
+
\smallgain\,\bar{M}\bar{a}c_h\,
\dist{\mathcal Q_{\mathrm{eq},k}}(\boldsymbol{q}_{k,0}).
\end{equation}
}%

Finally, since \(\boldsymbol{F}_k\) is uniformly bounded, there exists
\(\bar{F}_\Delta>0\) such that
{\small
\[
\|\boldsymbol{F}_{k+1}-\boldsymbol{F}_k\|\le \bar{F}_\Delta,
\qquad \forall k.
\]
}%
Using \eqref{eq:qref_bound} and \eqref{eq:ckdiff_bound} in
\eqref{eq:DeltaQ_pre} yields
{\small
\[
\Delta_{\mathcal Q,k}
\le
c_{Q,w}\,w_k
+
\smallgain\,c_{Q,q}\,\dist{\mathcal Q_{\mathrm{eq},k}}(\boldsymbol{q}_{k,0}),
\]
}%
which proves \eqref{eq:DeltaQ_bound_final_rewritten}.
\end{proof}

We can now combine the set-drift estimate with the contraction of the
frozen ADMM subsystem to derive a perturbed estimate for the fast
dynamics.

\begin{lemma}[Perturbed estimate for the ADMM subsystem]
\label{lem:fast_iss}
Define
{\small
\[
d_k \define \dist{\mathcal Q_{\mathrm{eq},k}}(\boldsymbol{q}_{k,0}).
\]
}%
Then
{\small
\begin{equation}
\label{eq:dq_iss_rewritten}
d_{k+1}
\le
(\mu^{\maxiteradmm}+\smallgain c_{Q,q})\,d_k
+
c_{Q,w}\,w_k,
\qquad \forall k.
\end{equation}
}%
\end{lemma}

\begin{proof}
By \eqref{eq:dk_triangle},
{\small
\[
d_{k+1}
\le
\dist{\mathcal Q_{\mathrm{eq},k}}(\boldsymbol q_{k+1,0})
+
\Delta_{\mathcal Q,k}.
\]
}%
Using the warm-start relation \(\boldsymbol q_{k+1,0}=\boldsymbol q_{k,\maxiteradmm}\),
and then Lemma~\ref{lem:bl}, we obtain
{\small
\[
\dist{\mathcal Q_{\mathrm{eq},k}}(\boldsymbol q_{k+1,0})
=
\dist{\mathcal Q_{\mathrm{eq},k}}(\boldsymbol q_{k,\maxiteradmm})
\le
\mu^{\maxiteradmm} d_k.
\]
}%
Combining the two inequalities with Lemma~\ref{lem:set_drift} yields
{\small
\[
d_{k+1}
\le
\mu^{\maxiteradmm}d_k
+
c_{Q,w}\,w_k
+
\smallgain c_{Q,q}\,d_k,
\]
}%
which proves \eqref{eq:dq_iss_rewritten}.
\end{proof}

We are now in a position to combine the perturbed estimates of the
reduced and ADMM dynamics into a comparison system for the full
interconnection.

\begin{proof}[Proof of Theorem~\ref{thm:main}]
By Lemmas~\ref{lem:slow_iss} and \ref{lem:fast_iss}, the interconnected
system satisfies
{\small
\[
\begin{aligned}
w_{k+1}
&\le
\sqrt{\ffactor}\,w_k
+
\smallgain c_{12}\,d_k,
\\
d_{k+1}
&\le
c_{21}\,w_k
+
(\mu^{\maxiteradmm}+\smallgain c_{22})\,d_k,
\end{aligned}
\]
}%
where, for compactness,
{\small
\[
c_{12}\define c_w,
\qquad
c_{21}\define c_{Q,w},
\qquad
c_{22}\define c_{Q,q}.
\]
}%
Equivalently,
{\small
\[
\begin{bmatrix}
w_{k+1}\\[1mm]
d_{k+1}
\end{bmatrix}
\le
\boldsymbol{\mathcal A}
\begin{bmatrix}
w_k\\[1mm]
d_k
\end{bmatrix},
\qquad
\boldsymbol{\mathcal A}
=
\begin{bmatrix}
\sqrt{\ffactor} & \smallgain c_{12}\\[1mm]
c_{21} & \mu^{\maxiteradmm}+\smallgain c_{22}
\end{bmatrix}.
\]
}%

Since \(\boldsymbol{\mathcal A}\in\mathbb R_{\ge 0}^{2\times 2}\), it is
enough to verify that
{\small
\[
\mathcal A_{11}<1,
\qquad
\mathcal A_{22}<1,
\qquad
\det(\eye{}-\boldsymbol{\mathcal A})>0.
\]
}%
Under these conditions, \(\eye{}-\boldsymbol{\mathcal A}\) is a
nonsingular \(M\)-matrix
\cite[Theorem~6.2.3, condition~(E17)]{berman1994nonnegative}, and
therefore \(\rho(\boldsymbol{\mathcal A})<1\)
\cite[Def.~6.1.2]{berman1994nonnegative}.

First,
{\small
\[
\mathcal A_{11}=\sqrt{\ffactor}<1.
\]
}%
Next, \eqref{eq:smallgain_bound} implies
{\small
\[
\smallgain
\bigl(
c_{12}c_{21}+(1-\sqrt{\ffactor})c_{22}
\bigr)
<
(1-\sqrt{\ffactor})(1-\mu^{\maxiteradmm}),
\]
}%
and therefore
{\small
\[
\smallgain c_{22}<1-\mu^{\maxiteradmm},
\]
}%
which yields
{\small
\[
\mathcal A_{22}=\mu^{\maxiteradmm}+\smallgain c_{22}<1.
\]
}%
Finally,
{\small
\[
\det(\eye{}-\boldsymbol{\mathcal A})
=
(1-\sqrt{\ffactor})(1-\mu^{\maxiteradmm})
-
\smallgain\bigl(
c_{12}c_{21}+(1-\sqrt{\ffactor})c_{22}
\bigr),
\]
}%
which is positive by \eqref{eq:smallgain_bound}.
Hence \(\boldsymbol{\mathcal A}\) is Schur.

Therefore, there exist constants \(C>0\) and \(\lambda\in(0,1)\) such that
{\small
\[
\|\boldsymbol{\mathcal A}^k\|\le C\lambda^k,
\qquad \forall k\ge 0.
\]
}%
Iterating the comparison inequality gives
{\small
\[
\begin{bmatrix}
w_k\\[1mm]
d_k
\end{bmatrix}
\le
\boldsymbol{\mathcal A}^k
\begin{bmatrix}
w_0\\[1mm]
d_0
\end{bmatrix},
\]
}%
and hence
{\small
\[
w_k+d_k
\le
C\lambda^k\bigl(w_0+d_0\bigr).
\]
}%

Finally, by Lemma~\ref{lemma:conditioning},
{\small
\[
\sqrt{\underline s_\smallgain}\,\|\estate_k\|
\le
w_k
\le
\sqrt{\bar s_\smallgain}\,\|\estate_k\|,
\]
}%
so \(w_k\) is uniformly equivalent to \(\|\estate_k\|\).
Recalling that
\(d_k=\dist{\mathcal Q_{\mathrm{eq},k}}(\boldsymbol{q}_{k,0})\),
we conclude that
{\small
\[
\|\estate_k\|
+
\dist{\mathcal Q_{\mathrm{eq},k}}(\boldsymbol{q}_{k,0})
\le
C\lambda^k
\bigl(
\|\estate_0\|
+
\dist{\mathcal Q_{\mathrm{eq},0}}(\boldsymbol q_{0,0})
\bigr),
\]
}%
which proves the claim.
\end{proof}

\begin{remark}[Non-circularity of condition~\eqref{eq:smallgain_bound}]
\label{rem:noncircular}
Although \(c_{12}\) depends on \(\smallgain\) through \(\bar{s}_\smallgain\),
the coefficient \(c_{21}=c_{Q,w}\) contains the reciprocal factor
\((\underline{s}_\smallgain)^{-1/2}\), so the product \(c_{12}c_{21}\) is
independent of \(\smallgain\). Indeed, by Lemma~\ref{lemma:conditioning},
\[
\bar{s}_\smallgain=\smallgain\,\bar s,
\qquad
\underline{s}_\smallgain=\smallgain\,\underline s.
\]
Hence
{\small
\[
c_{12}c_{21}=C_0,
\]
}%
for a constant \(C_0\) independent of \(\smallgain\). Condition
\eqref{eq:smallgain_bound} is therefore equivalent to a
\(\smallgain\)-independent small-gain requirement.
\end{remark}

\begin{remark}
The final Schur condition on \(\boldsymbol{\mathcal A}\) captures the
required separation between the slow and fast dynamics. In practice,
this condition can be improved by decreasing \(\smallgain\), tuning the
ADMM parameters \((\rho,\alpha)\), or by increasing the number
\(\maxiteradmm\) of ADMM iterations performed per observer step.
\end{remark}


\section{Simulation results}
We evaluate the proposed distributed observer in a cooperative localization task with
$N=10$ agents connected by a randomly generated proximity communication graph.
A number $M=3$ of agents are anchors with access to absolute position measurements, while all other agents only use relative measurements.
The robots state is described by
$
\state_{i,k}
=
\begin{bmatrix}
\pos_{i,k}^\top & \vel_{i,k}^\top
\end{bmatrix}^\top
\in\nR{4}
$, where $\pos_{i,k} \in \nR{2}$ denotes the position of the $i$-th robot, $\vel_{i,k} \in \nR{2}$ its velocity, the input is given by the acceleration $\sysinput_{i,k}\in\nR{2}$, and
{\small\[
\statemat_i =
\begin{bmatrix}
1 & T_s \\
0 & 1 
\end{bmatrix}\otimes \eye{2}
,
\qquad
\inputmat_i =
\begin{bmatrix}
\frac{1}{2}T_s^2 \\
T_s  
\end{bmatrix}\otimes \eye{2}.
\]}%
where $T_s=0.05$ is the sampling time. 

Measurements are
\[
\meas_{i,k}^{\ell}=\pos_{i,k}\quad (i\in\vertices_p),\qquad
\meas_{ij,k}^r=\pos_{i,k}-\pos_{j,k}\quad (i,j)\in\mathcal{E}.
\]
Initial positions are drawn uniformly in the workspace and initial velocities are random.
The initial state estimate is perturbed by
$\stateest_{0|0}=\state_0+\boldsymbol{\nu}$, with
$\boldsymbol{\nu}\sim\mathcal{N}(0,\boldsymbol{P}_0),\;
\boldsymbol{P}_0=\diag(1,1,0.1,0.1)\otimes I_{\Nrob}$.
While, for simplicity, a scalar $\ffactor$ is used in the proof, a diagonal forgetting matrix $\ffmatrix$ can help improve performance. Here we use a diagonal forgetting matrix with larger decay on velocity than position:
$\gamma_p=e^{-5T_s}\approx0.779$,
$\gamma_v=e^{-50T_s}\approx0.082$,
applied as $\ffmatrix_i=\diag(\gamma_p,\gamma_p,\gamma_v,\gamma_v)$. Moreover, we use a correction gain $\smallgain = 1$, and
the information-weight gains are chosen as
$(\outcov_i^{\ell})^{-1}=5I_2$, $(\outcov_{ij}^r)^{-1}=0.5I_2$.

We compare the centralized Kalman-like observer with three distributed correction schemes, namely a Richardson-based implementation and the two proposed ADMM-based implementation, i.e. the one in Algorithm.~\ref{alg:dist_observer_short} simply indicated as ADMM in the legend and the modification described in Sec.~\ref{subsec:private_split}, indicated as ADMM-direct. {The number of inner loop iterations for both Richardson and ADMM is fixed to the same value $H=1$ and we use $\smallgain=1$.} For the Richardson method, we use the fixed gain $\alpha_R=0.05$. Because the information matrix $\infomat_{k|k}$ is time-varying, the optimal Richardson step size
{\small
\[
\alpha_R^\ast=\frac{2}{\lambda_{\min}(\infomat_{k|k})+\lambda_{\max}(\infomat_{k|k})}
\]
}%
also varies with time. The value $\alpha_R=0.05$ was thus chosen as a fixed compromise, based on simulation data: it remains close to the optimal value over most of the time horizon, as estimated from the average values of $\lambda_{\min}(\infomat_{k|k})$ and $\lambda_{\max}(\infomat_{k|k})$, while preserving stability of the iteration.
For the ADMM-based method, we use the fixed penalty and relaxation parameters $(\rho,\alpha)=(1,0.95)$, selected empirically through simulation. In all methods, the prediction and measurement-update steps coincide; only the correction step is implemented differently.

Figure~\ref{fig:init_config} shows the initial sensing graph, together with the true and estimated states of the agents.
\begin{figure}[h]
    \centering
\includegraphics[width=0.6\linewidth]{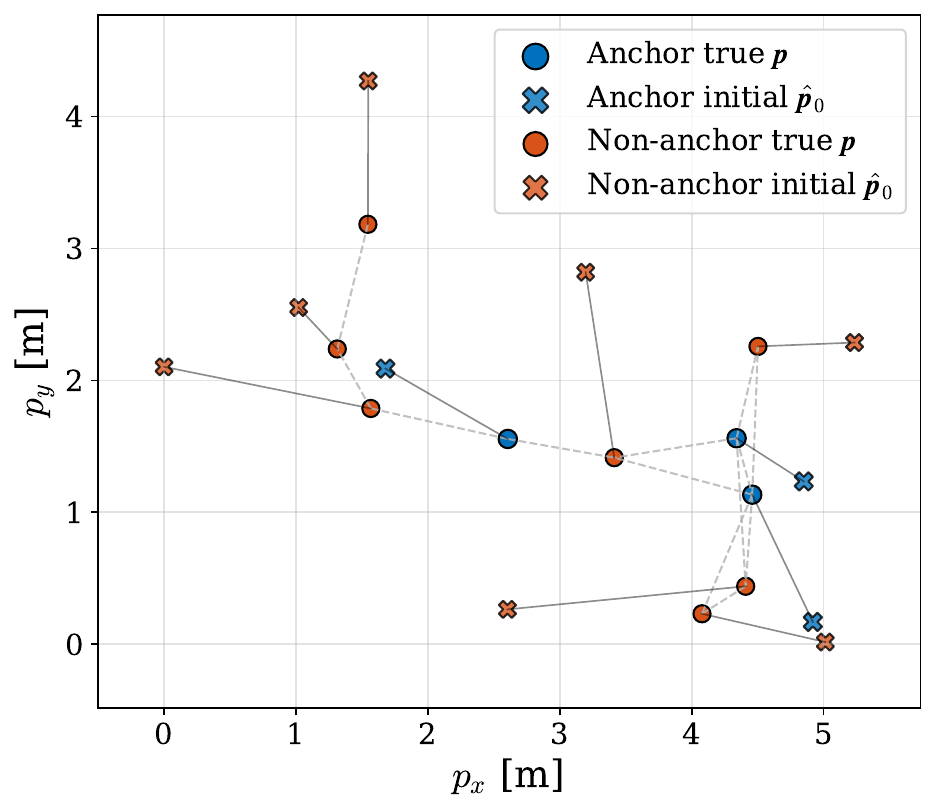}
    \caption{Initial robot configuration and sensing graph. Circle markers denote the true robot positions, while cross markers denote the corresponding initial estimated positions. The anchor robots are in blue, and the dashed green lines represent the graph edges, while solid lines connect a robot true position to the corresponding estimate.}
    \label{fig:init_config}
\end{figure}
Fig.~\ref{fig:est_corr_error}a and the zoomed in Fig.~\ref{fig:est_corr_error1000} report the global estimation error norm, showing that both distributed schemes track the centralized estimate closely, with ADMM achieving closer performance to the centralized scheme, as visible in Fig.~\ref{fig:est_corr_error1000}.

\begin{figure}[h]
    \centering
\includegraphics[width=1.0\linewidth]{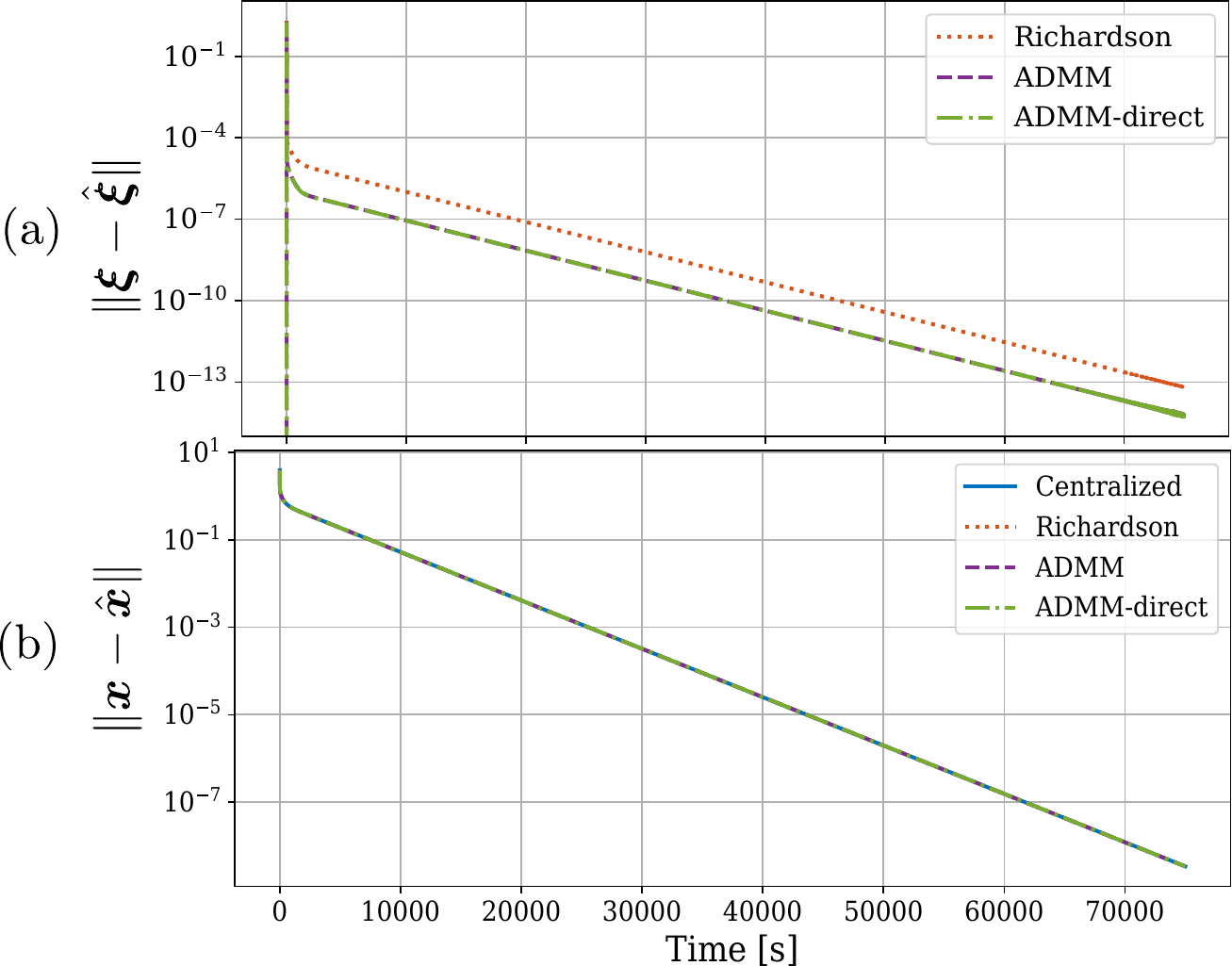}
    \caption{Results over the full simulation. (a) Norm of the global estimation error, i.e. $\|\estate_k\|$. (b) Norm of the global correction error, i.e. $\|\ecorr_k\|$.}
    \label{fig:est_corr_error}
\end{figure}

\begin{figure}[h]
    \centering
\includegraphics[width=1.0\linewidth]{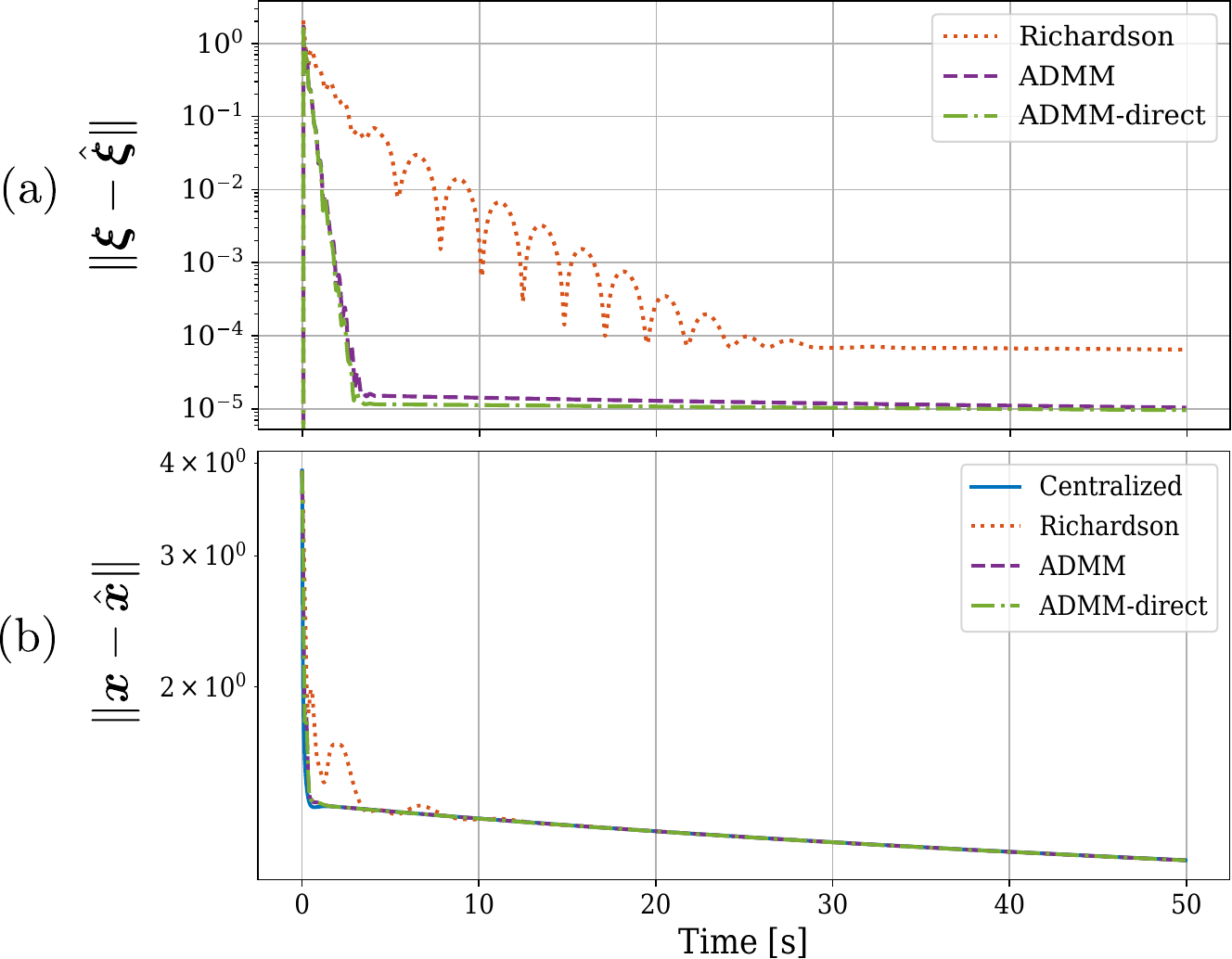}
    \caption{Zoom of the results over the first 1000 steps. (a) Norm of the global estimation error, i.e. $\|\estate_k\|$. (b) Norm of the global correction error, i.e. $\|\ecorr_k\|$.}
    \label{fig:est_corr_error1000}
\end{figure}

Fig.~\ref{fig:est_corr_error}b and the zoomed in  Fig.~\ref{fig:est_corr_error} show the global correction error norm with respect to the distributed Richardson, indicating that ADMM achieves considerably faster convergence to the centralized correction term under the selected tuning. A further improvement is visible when using the modification described in Sec.~\ref{subsec:private_split}.

Overall, the proposed ADMM-based distributed observer closely reproduces centralized performance while preserving sparse local communication.


\section{Conclusions}\label{sec:conclusions}

This work proposed a distributed Kalman-like observer for multi-agent systems with local and relative measurements. The approach combines a sparsity-preserving information prediction step based on exponential forgetting with a distributed correction step obtained by solving the resulting sparse linear system through partition-based ADMM. In contrast to centralized information-form filtering, the proposed construction preserves the graph-induced sparsity of the information matrix and avoids both dense Riccati-type prediction updates and centralized matrix inversion.

A  two-time-scale separation was used to separate the reduced estimation dynamics from the distributed optimization layer, and small gain arguments established uniform global exponential stability of the interconnected observer.

The numerical results in a cooperative localization scenario showed that the proposed ADMM-based correction closely reproduces the centralized observer while converging significantly faster than the previously considered Richardson-based dynamic inversion scheme. 

Future work will consider asynchronous and event-triggered implementations, switching graphs, robustness with respect to communication losses, and packet drops, and extensions to nonlinear estimation problems.


\appendix


\subsection{Derivation of the correction form}\label{app:kalupdate}

From the standard information-form update
{\small
\begin{equation}
\label{eq:app_info_update}
\infomat_{k|k}
=
\infomat_{k|k-1}
+
\smallgain\,\measmat_k^\top \outcov_k \measmat_k,
\qquad
\infovec_{k|k}
=
\infovec_{k|k-1}
+
\smallgain\,\measmat_k^\top \outcov_k \meas_k,
\end{equation}
}%
together with
{\small
$
\infovec_{k|k-1}=\infomat_{k|k-1}\stateest_{k|k-1}
$
}%
and
{\small
$
\stateest_{k|k}=\infomat_{k|k}^{-1}\infovec_{k|k},
$
}%
it follows that
{\small
\begin{equation}\label{eq:xhatkkapp}
\stateest_{k|k}
=
\infomat_{k|k}^{-1}
\Bigl(
\infomat_{k|k-1}\stateest_{k|k-1}
+
\smallgain\,\measmat_k^\top \outcov_k \meas_k
\Bigr).
\end{equation}
}%
Expanding {\small$\infomat_{k|k}$} as in \eqref{eq:app_info_update}, we also have
{\small
\[
\infomat_{k|k}\stateest_{k|k-1}
=
\infomat_{k|k-1}\stateest_{k|k-1}
+
\smallgain\,\measmat_k^\top \outcov_k \measmat_k\stateest_{k|k-1}.
\]
}%
After pre-multiplying both sides by {\small$\infomat_{k|k}^{-1}$}, subtracting this expression from the \eqref{eq:xhatkkapp} yields
{\small
\[
\stateest_{k|k}-\stateest_{k|k-1}
=
\smallgain\infomat_{k|k}^{-1}
\,\measmat_k^\top \outcov_k
\bigl(
\meas_k-\measmat_k\stateest_{k|k-1}
\bigr).
\]
}%
Therefore, defining
{\small
\[
\correction_k
:=
\frac{1}{\smallgain}
\bigl(\stateest_{k|k}-\stateest_{k|k-1}\bigr),
\qquad
\knownterm_k
:=
\measmat_k^\top \outcov_k
\bigl(\meas_k-\measmat_k\stateest_{k|k-1}\bigr),
\]
}%
we obtain the correction form
{\small
\[
\stateest_{k|k}
=
\stateest_{k|k-1}
+
\smallgain\,\correction_k,
\qquad
\infomat_{k|k}\correction_k=\knownterm_k.
\]
}%


\subsection{UGES of the reduced system with matrix forgetting factor}\label{app:matrixffuges}

In the case of a matrix forgetting factor $\ffmatrix$, a corresponding UGES result can still be established. However, unlike the scalar case, the resulting contraction condition is difficult to verify in practice and generally very conservative.

\begin{lemma}[UGES of the reduced system with matrix forgetting factor]
\label{lem:reduced_uges_matrix}
Under Assumptions~\ref{ass:boundedness}, \ref{ass:compobs}, and
\ref{ass:invertibility}, consider the reduced system~\eqref{eq:reduced_system}
with the matrix-forgetting prediction \eqref{eq:forgetting_prediction_matrix}.
Suppose that there exists $\bar\gamma\in(0,1)$ such that
{\small
\begin{equation}
\label{eq:matrix_contraction_assumption}
    \ffmatrix^\top \infomat_{k|k}\ffmatrix
    \preceq
    \bar\gamma\,\infomat_{k|k},
    \qquad \forall\,k.
\end{equation}
}
Then the reduced system is uniformly globally exponentially stable, with
Lyapunov function {\small
$V_k \define \estate_k^\top \infomat_{k|k-1}\estate_k$}
satisfying
{\small
\[
    \underline{s}^\smallgain\|\estate_k\|^2
    \le
    V_k
    \le
    \bar{s}^\smallgain\|\estate_k\|^2,
    \qquad
    V_{k+1}-V_k
    \le
    -(1-\bar\gamma)\underline{s}^\smallgain\|\estate_k\|^2.
\]
}
 
Moreover, condition~\eqref{eq:matrix_contraction_assumption} admits the
following equivalent and sufficient characterizations, in decreasing order
of tightness:
\begin{enumerate}[(\roman*)]
\item \emph{(Equivalent condition)}
    \eqref{eq:matrix_contraction_assumption} holds if and only if
    {\small
    \[
        \bigl\|\infomat_{k|k}^{1/2}\ffmatrix\infomat_{k|k}^{-1/2}\bigr\|_2^2
        \le \bar\gamma,
        \qquad \forall\,k.
    \]
    }
\item \emph{(Sufficient condition)}
    A sufficient condition is
    {\small
    \[
        \|\ffmatrix\|_2^2\,\kappa(\infomat_{k|k})
        \le
        \bar\gamma
        < 1,
        \qquad \forall\,k,
    \]
    }
    where {\small$\kappa(\infomat_{k|k}) \define
    \lambda_{\max}(\infomat_{k|k})/\lambda_{\min}(\infomat_{k|k})$}
    is the condition number of $\infomat_{k|k}$.
\item \emph{(Uniform sufficient condition)}
    Using the uniform bounds of Lemma~\ref{lemma:conditioning}, a
    condition that is uniform in $k$ and independent of the trajectory is
    {\small
    \[
        \|\ffmatrix\|_2^2\,\frac{\bar{s}^\smallgain}{\underline{s}^\smallgain}
        < 1.
    \]
    }
\end{enumerate}
\end{lemma}
 
\begin{proof}
\textit{Lyapunov bound.}
From the reduced dynamics~\eqref{eq:reduced_system} and the matrix-forgetting prediction \eqref{eq:forgetting_prediction_matrix},
{\small
\[
    V_{k+1}
    =
    \estate_k^\top
    \infomat_{k|k-1}\infomat_{k|k}^{-1}
    \ffmatrix^\top \infomat_{k|k}\ffmatrix
    \infomat_{k|k}^{-1}\infomat_{k|k-1}
    \estate_k.
\]
}
Applying~\eqref{eq:matrix_contraction_assumption} and using
{\small
$\infomat_{k|k}\succeq\infomat_{k|k-1}$} (so that
{\small
$\infomat_{k|k-1}\infomat_{k|k}^{-1}\infomat_{k|k-1}\preceq\infomat_{k|k-1}$}),
{\small
\[
    V_{k+1} \le \bar\gamma\,V_k,
\]
}
and the stated decay follows from Lemma~\ref{lemma:conditioning}.
 
\textit{Equivalence of~(i).}
Pre- and post-multiplying~\eqref{eq:matrix_contraction_assumption}
by $\infomat_{k|k}^{-1/2}$ yields
{\small
\[
    \boldsymbol{\Upsilon}_k^\top\boldsymbol{\Upsilon}_k
    \preceq \bar\gamma I,
    \qquad
    \boldsymbol{\Upsilon}_k
    \define
    \infomat_{k|k}^{1/2}\ffmatrix\infomat_{k|k}^{-1/2},
\]
}
which is equivalent to $\|\boldsymbol{\Upsilon}_k\|_2^2\le\bar\gamma$.
 
\textit{Sufficiency of~(ii).}
Since {\small
$\ffmatrix^\top\infomat_{k|k}\ffmatrix
\preceq \|\ffmatrix\|_2^2\lambda_{\max}(\infomat_{k|k})I
\preceq \|\ffmatrix\|_2^2\kappa(\infomat_{k|k})\infomat_{k|k}$},
condition~\eqref{eq:matrix_contraction_assumption} holds whenever
{\small
$\|\ffmatrix\|_2^2\kappa(\infomat_{k|k})\le\bar\gamma$}.
 
\textit{Sufficiency of~(iii).}
The uniform bounds give
$\kappa(\infomat_{k|k})\le\bar{s}^\smallgain/\underline{s}^\smallgain$
for all $k$, so~(ii) is satisfied uniformly whenever
$\|\ffmatrix\|_2^2\,\bar{s}^\smallgain/\underline{s}^\smallgain<1$.
\end{proof}



\bibliographystyle{IEEEtran}
\bibliography{references}

@article{de2025distributed,
  title={A Distributed Kalman-like Observer with Dynamic Inversion-Based Correction for Multi-Agent Estimation},
  author={De Carli, Nicola and Dimarogonas, Dimos V},
  journal={IEEE Control Systems Letters},
  year={2025},
  publisher={IEEE}
}

@article{neal2011distributed,
  title={Distributed optimization and statistical learning via the alternating direction method of multipliers},
  author={Neal, Parikh and Eric, Chu and Borja, Peleato and Jonathan, Eckstein},
  journal={Foundations and Trends{\textregistered} in Machine learning},
  volume={3},
  number={1},
  pages={1--122},
  year={2011},
  publisher={Emerald Publishing Limited}
}

@article{bastianello2020asynchronous,
  title={Asynchronous distributed optimization over lossy networks via relaxed ADMM: Stability and linear convergence},
  author={Bastianello, Nicola and Carli, Ruggero and Schenato, Luca and Todescato, Marco},
  journal={IEEE Transactions on Automatic Control},
  volume={66},
  number={6},
  pages={2620--2635},
  year={2020},
  publisher={IEEE}
}

@article{bastianello2024robust,
  title={Robust online learning over networks},
  author={Bastianello, Nicola and Deplano, Diego and Franceschelli, Mauro and Johansson, Karl H},
  journal={IEEE Transactions on Automatic Control},
  volume={70},
  number={2},
  pages={933--946},
  year={2024},
  publisher={IEEE}
}

@inproceedings{bastianello2018partition,
  title={A partition-based implementation of the relaxed ADMM for distributed convex optimization over lossy networks},
  author={Bastianello, Nicola and Carli, Ruggero and Schenato, Luca and Todescato, Marco},
  booktitle={2018 IEEE Conference on Decision and Control (CDC)},
  pages={3379--3384},
  year={2018},
  organization={IEEE}
}

@inproceedings{shorinwa2020scalable,
  title={Scalable distributed optimization with separable variables in multi-agent networks},
  author={Shorinwa, Olaoluwa and Halsted, Trevor and Schwager, Mac},
  booktitle={2020 American Control Conference (ACC)},
  pages={3619--3626},
  year={2020},
  organization={IEEE}
}

@article{carnevale2025admm,
  title={ADMM-tracking gradient for distributed optimization over asynchronous and unreliable networks},
  author={Carnevale, Guido and Bastianello, Nicola and Notarstefano, Giuseppe and Carli, Ruggero},
  journal={IEEE Transactions on Automatic Control},
  volume={70},
  number={8},
  pages={5160--5175},
  year={2025},
  publisher={IEEE}
}

@article{carnevale2025unifying,
  title={A unifying system theory framework for distributed optimization and games},
  author={Carnevale, Guido and Mimmo, Nicola and Notarstefano, Giuseppe},
  journal={IEEE Transactions on Automatic Control},
  year={2025},
  publisher={IEEE}
}

@article{cticlea2013exponential,
  title={Exponential forgetting factor observer in discrete time},
  author={{\c{T}}iclea, Alexandru and Besan{\c{c}}on, Gildas},
  journal={Systems \& Control Letters},
  volume={62},
  number={9},
  pages={756--763},
  year={2013},
  publisher={Elsevier}
}

@article{bernard2020semi,
  title={On the semi-global stability of an ek-like filter},
  author={Bernard, Pauline and Mimmo, Nicola and Marconi, Lorenzo},
  journal={IEEE Control Systems Letters},
  volume={5},
  number={5},
  pages={1771--1776},
  year={2020},
  publisher={IEEE}
}

@article{gomez2011taxonomy,
  title={A taxonomy of multi-area state estimation methods},
  author={G{\'o}mez-Exp{\'o}sito, Antonio and De La Villa Ja{\'e}n, Antonio and G{\'o}mez-Quiles, Catalina and Rousseaux, Patricia and Van Cutsem, Thierry},
  journal={Electric Power Systems Research},
  volume={81},
  number={4},
  pages={1060--1069},
  year={2011},
  publisher={Elsevier}
}

@article{yang2021distributed,
  title={Distributed Kalman-like filtering and bad data detection in the large-scale power system},
  author={Yang, Jun and Zhang, Wen-An and Guo, Fanghong},
  journal={IEEE transactions on industrial informatics},
  volume={18},
  number={8},
  pages={5096--5104},
  year={2021},
  publisher={IEEE}
}

@article{luft2018recursive,
  title={Recursive decentralized localization for multi-robot systems with asynchronous pairwise communication},
  author={Luft, Lukas and Schubert, Tobias and Roumeliotis, Stergios I and Burgard, Wolfram},
  journal={The International Journal of Robotics Research},
  volume={37},
  number={10},
  pages={1152--1167},
  year={2018},
  publisher={SAGE Publications Sage UK: London, England}
}

@article{de2024adaptive,
  title={Adaptive observer from body-frame relative position measurements for cooperative localization},
  author={De Carli, Nicola and Restrepo, Esteban and Giordano, Paolo Robuffo},
  journal={IEEE Control Systems Letters},
  volume={8},
  pages={1337--1342},
  year={2024},
  publisher={IEEE}
}

@article{chang2021resilient,
  title={Resilient and consistent multirobot cooperative localization with covariance intersection},
  author={Chang, Tsang-Kai and Chen, Kenny and Mehta, Ankur},
  journal={IEEE Transactions on Robotics},
  volume={38},
  number={1},
  pages={197--208},
  year={2021},
  publisher={IEEE}
}

@article{kia2016cooperative,
  title={Cooperative localization for mobile agents: A recursive decentralized algorithm based on Kalman-filter decoupling},
  author={Kia, Solmaz S and Rounds, Stephen and Martinez, Sonia},
  journal={IEEE Control Systems Magazine},
  volume={36},
  number={2},
  pages={86--101},
  year={2016},
  publisher={IEEE}
}

@inproceedings{jung2020decentralized,
  title={Decentralized collaborative state estimation for aided inertial navigation},
  author={Jung, Roland and Brommer, Christian and Weiss, Stephan},
  booktitle={2020 IEEE International Conference on Robotics and Automation (ICRA)},
  pages={4673--4679},
  year={2020},
  organization={IEEE}
}

@inproceedings{de2021online,
  title={Online decentralized perception-aware path planning for multi-robot systems},
  author={De Carli, Nicola and Salaris, Paolo and Giordano, Paolo Robuffo},
  booktitle={2021 International Symposium on Multi-Robot and Multi-Agent Systems (MRS)},
  pages={128--136},
  year={2021},
  organization={IEEE}
}

@article{roumeliotis2002distributed,
  title={Distributed multirobot localization},
  author={Roumeliotis, Stergios I and Bekey, George A},
  journal={IEEE transactions on robotics and automation},
  volume={18},
  number={5},
  pages={781--795},
  year={2002},
  publisher={IEEE}
}

@inproceedings{arambel2001covariance,
  title={Covariance intersection algorithm for distributed spacecraft state estimation},
  author={Arambel, Pablo O and Rago, Constantino and Mehra, Raman K},
  booktitle={Proceedings of the 2001 American Control Conference.(Cat. No. 01CH37148)},
  volume={6},
  pages={4398--4403},
  year={2001},
  organization={IEEE}
}

@inproceedings{carrillo2013decentralized,
  title={Decentralized multi-robot cooperative localization using covariance intersection},
  author={Carrillo-Arce, Luis C and Nerurkar, Esha D and Gordillo, Jos{\'e} L and Roumeliotis, Stergios I},
  booktitle={2013 IEEE/RSJ international conference on intelligent robots and systems},
  pages={1412--1417},
  year={2013},
  organization={IEEE}
}

@article{klingner2019fault,
  title={Fault-tolerant covariance intersection for localizing robot swarms},
  author={Klingner, John and Ahmed, Nisar and Correll, Nikolaus},
  journal={Robotics and Autonomous Systems},
  volume={122},
  pages={103306},
  year={2019},
  publisher={Elsevier}
}

@article{battistelli2014consensus,
  title={Consensus-based linear and nonlinear filtering},
  author={Battistelli, Giorgio and Chisci, Luigi and Mugnai, Giovanni and Farina, Alfonso and Graziano, Antonio},
  journal={IEEE Transactions on Automatic Control},
  volume={60},
  number={5},
  pages={1410--1415},
  year={2014},
  publisher={IEEE}
}

@article{kamal2013information,
  title={Information weighted consensus filters and their application in distributed camera networks},
  author={Kamal, Ahmed T and Farrell, Jay A and Roy-Chowdhury, Amit K},
  journal={IEEE Transactions on Automatic Control},
  volume={58},
  number={12},
  pages={3112--3125},
  year={2013},
  publisher={IEEE}
}

@article{sebastian2023eco,
  title={ECO-DKF: Event-triggered and certifiable optimal distributed Kalman filter under unknown correlations},
  author={Sebasti{\'a}n, Eduardo and Montijano, Eduardo and Sag{\"u}{\'e}s, Carlos},
  journal={IEEE Transactions on Automatic Control},
  volume={69},
  number={4},
  pages={2613--2620},
  year={2023},
  publisher={IEEE}
}

@book{bertsekas2015parallel,
  title={Parallel and distributed computation: numerical methods},
  author={Bertsekas, Dimitri and Tsitsiklis, John},
  year={2015},
  publisher={Athena Scientific}
}

@inproceedings{de2024distributed,
  title={Distributed control barrier functions for global connectivity maintenance},
  author={De Carli, Nicola and Salaris, Paolo and Giordano, Paolo Robuffo},
  booktitle={2024 IEEE International Conference on Robotics and Automation (ICRA)},
  pages={12048--12054},
  year={2024},
  organization={IEEE}
}

@article{parikh2014proximal,
  title={Proximal algorithms},
  author={Parikh, Neal and Boyd, Stephen},
  journal={Foundations and Trends in optimization},
  volume={1},
  number={3},
  pages={127--239},
  year={2014},
  publisher={Emerald Publishing Limited}
}

@book{bauschke2011convex,
  author    = {Bauschke, Heinz H. and Combettes, Patrick L.},
  title     = {Convex Analysis and Monotone Operator Theory in Hilbert Spaces},
  publisher = {Springer},
  year      = {2011},
  series    = {CMS Books in Mathematics},
  address   = {New York, NY},
  doi       = {10.1007/978-1-4419-9467-7},
  isbn      = {978-1-4419-9467-7},
  edition   = {1st}
}

@article{ryu2016primer,
  title={Primer on monotone operator methods},
  author={Ryu, Ernest K and Boyd, Stephen},
  journal={Appl. comput. math},
  volume={15},
  number={1},
  pages={3--43},
  year={2016}
}

@book{berman1994nonnegative,
  title={Nonnegative matrices in the mathematical sciences},
  author={Berman, Abraham and Plemmons, Robert J},
  year={1994},
  publisher={SIAM}
}

\end{document}